\documentclass[12pt]{article}
\usepackage{natbib,rotating,subcaption}

\usepackage{amsthm,amsmath,natbib,algorithm,algpseudocode,pifont,lscape,varwidth}

\RequirePackage[colorlinks,citecolor=blue,urlcolor=blue]{hyperref}
\usepackage{mathtools}
\usepackage{amssymb}
\usepackage[T1]{fontenc}
\usepackage[utf8x]{inputenc}
\usepackage{commath,colonequals}
\usepackage{enumitem,graphicx}
\usepackage{multirow,xcolor}
\usepackage[normalem]{ulem}

\newcommand{\sampcov}{\mathbf{S}}

\newcommand{\vecx}{\mathbf{x}}

\newcommand{\vecX}{\mathbf{X}}

\newcommand{\vecQ}{\mathbf{Q}}

\newcommand{\vecw}{\mathbf{w}}
\newcommand{\vecz}{\mathbf{z}}

\newcommand{\vecmu}{\boldsymbol\mu}

\newcommand{\vecLambda}{\mathbf\Lambda}

\newcommand{\vecSigma}{\mathbf\Sigma}

\newcommand{\vectheta}{\boldsymbol\theta}

\newcommand{\vecvartheta}{\boldsymbol\vartheta}

\newcommand{\argmax}{\arg\max}
\DeclareUnicodeCharacter{87}{\u{o}}

\newtheorem{assumption}{Assumption}
\newtheorem{lemma}{Lemma}
\newtheorem{theorem}{Theorem}
\DeclarePairedDelimiter\abss{\lvert}{\rvert}

\theoremstyle{rmk}
\newtheorem{rmk}{Remark}
\newtheorem{definition}{Definition}

\title{Finding outliers in Gaussian model-based clustering}
\author{Katharine M. Clark and Paul D. McNicholas}
\date{\small Department of Mathematics \& Statistics, McMaster University, Ontario, Canada.}

\begin{document}
\maketitle

\begin{abstract}
Clustering, or unsupervised classification, is a task often plagued by outliers. Yet there is a paucity of work on handling outliers in clustering. Outlier identification algorithms tend to fall into three broad categories: outlier inclusion, outlier trimming, and \textit{post hoc} outlier identification methods, with the former two often requiring pre-specification of the number of outliers. The fact that sample squared Mahalanobis distance is beta-distributed is used to derive an approximate distribution for the log-likelihoods of subset finite Gaussian mixture models. An algorithm is then proposed that removes the least plausible points according to the subset log-likelihoods, which are deemed outliers, until the subset log-likelihoods adhere to the reference distribution. This results in a trimming method, called OCLUST, that inherently estimates the number of outliers.\\[-10pt]

\end{abstract}

\section{Introduction}
Classification aims to partition data into a number of groups, or classes, such that observations in the same group are in some sense similar to one another and different to observations in other groups. Clustering is unsupervised classification, in that none of the group memberships are known \textit{a priori}. Many clustering algorithms originate from one of three major methods: hierarchical clustering, $k$-means/medoids clustering, and mixture model-based clustering. Additional popular methods include spectral and density-based clustering. Hierarchical clustering either iteratively merges (agglomerative) or splits (divisive) clusters, while $k$-means/medoids aims to minimize distance from each point to its cluster centre. Spectral clustering is based in graph theory, where a similarity matrix is used to map nearby points into a lower dimension. The mapped points are then clustered by distance with traditional algorithms, e.g., $k$-means. Density-based methods, such as DBSCAN \citep{ester96}, cluster based on number of points within a certain proximity to other points. Although distance-based clustering remains popular, the mixture modelling approach has become increasingly prevalent due to its robustness and mathematical interpretability \citep{mcnicholas16a,mcnicholas16b}. Typically, in the mixture modelling framework for clustering, each component corresponds to a unique cluster, and a cluster is a sample of points from the component distribution. Although the model can employ almost any component distribution, Gaussian components have been popular due their tractability. Most mixture model-based clustering methods assume, either explicitly or implicitly, that the data --- and the clusters --- are free of outliers.

Mixture model-based clustering involves maximizing the likelihood of the mixture model. The density of a $G$-component Gaussian mixture model is a convex linear combination of  component densities and is given by
\begin{equation}
	f(\vecx\mid\vecvartheta)=\sum_{g=1}^{G}\pi_g \phi(\vecx\mid\vecmu_g,\vecSigma_g),
	\label{eq:generalmixture1}
\end{equation}
where
\begin{equation*}
	\phi(\vecx\mid\vecmu_g,\vecSigma_g)=\frac{1}{\sqrt{(2\pi)^p|\vecSigma_g|}}\text{exp}\left\{-\frac{1}{2}(\vecx-\vecmu_g)'\vecSigma_g^{-1}(\vecx-\vecmu_g)\right\}
\end{equation*}
is the density of a $p$-dimensional random variable $\vecX$ from a Gaussian distribution with mean $\vecmu_g$ and covariance matrix $\vecSigma_g$, $\pi_g>0$ is the mixing proportion such that $\sum_{g=1}^G \pi_g =1$, and $\vecvartheta=\{\pi_1, \dots, \pi_G, \vecmu_1, \dots, \vecmu_G, \vecSigma_1, \dots, \vecSigma_G\}$ denotes the parameters.

An outlier can be considered an observation ``that appears to deviate markedly from other members of the sample in which it occurs'' \citep{grubbs69}. The two main types of outliers are mild outliers and gross outliers \citep[][pp.~79--80]{ritter14}. Mild outliers may be near other points in a cluster but they are unusual relative to the distribution of the cluster. Gross outliers are unpredictable and are not modeled by any probability distribution. They do not exist in close proximity to any of the clusters. Outliers may occur due to (unlikely) random chance, or they may arise due to experimental, measurement, or recording error \citep{grubbs69}.

A noisy dataset was simulated by generating three Gaussian clusters and adding uniform noise, and is shown on the left-hand side of Figure~\ref{fig:toy}. The uniform noise can be thought of as representing mostly gross or mild outliers; however, some of the uniform noise overlaps with a cluster and so would be classified as belonging to the cluster by any sensible method. Applying the \texttt{mixture} \citep{pocuca21} package for \textsf{R} \citep{R21} to these data, with $G=3$ components, leads to the result depicted on the right-hand side of Figure~\ref{fig:toy}. This solution merges two of the clusters  and places some outliers into their own cluster. Therefore, the result of failing to account for outliers is an ill-fitting model which misrepresents the structure of the data.
\begin{figure}[!htb]
	\centering
	\includegraphics[width=6in]{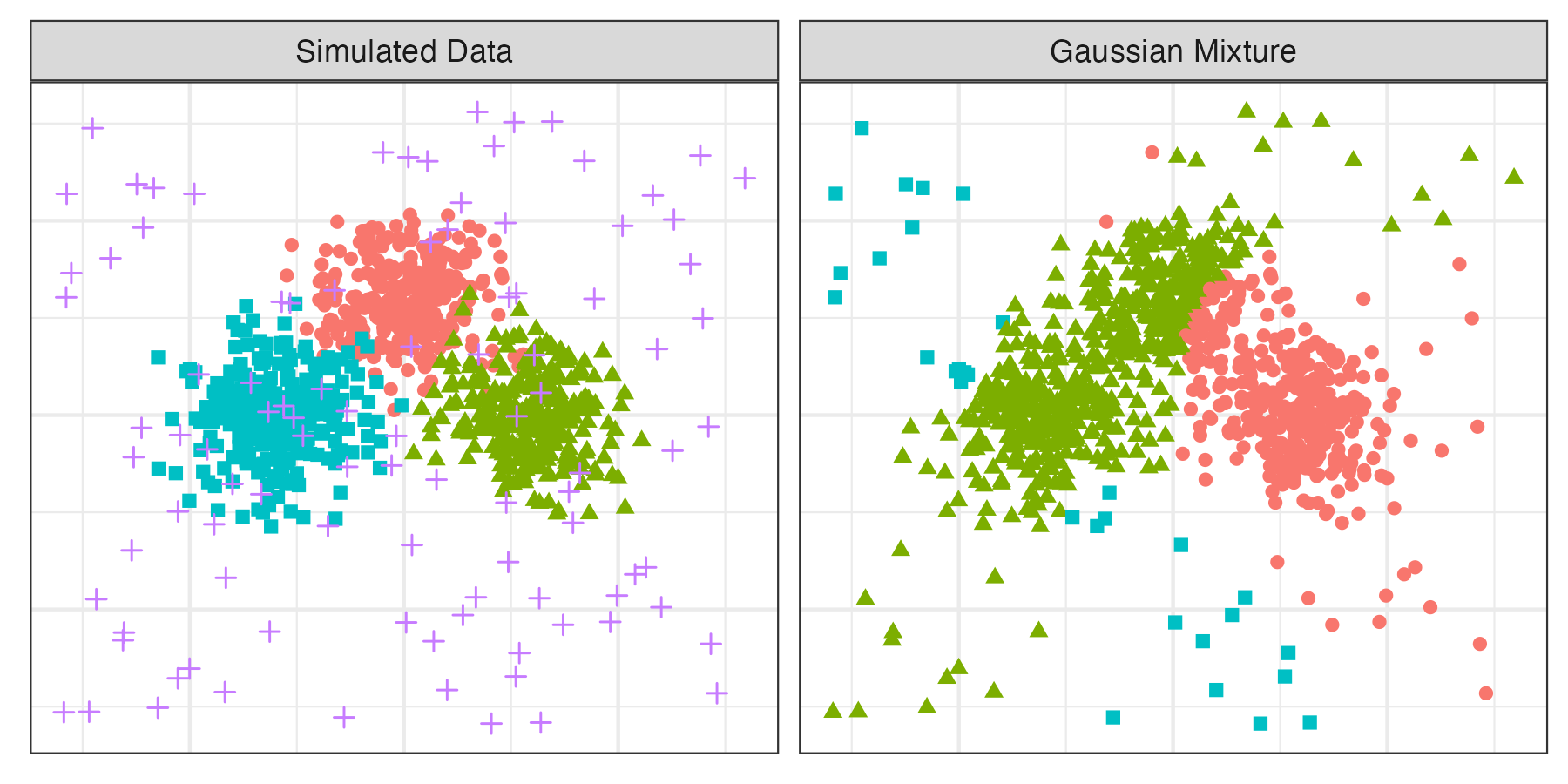}
	\caption{A simulated dataset with three Gaussian clusters and uniform noise. Colours represent the simulated classes (left) and predicted classes using the \texttt{mixture} package (right).}
	\label{fig:toy}
\end{figure}

In general, outliers, particularly gross ones, can significantly impact the parameter estimates. For example, when fitting a mixture model, outliers may draw component means towards them and inflate component covariance matrices. In the case of our simulated example, the originally-spherical clusters are elongated (Figure~\ref{fig:toy}). It is thus beneficial to either remove outliers or reduce their effect by accounting for them. The proposed OCLUST algorithm is designed to handle both mild and gross outliers.

\section{Related Work}
\subsection{Outliers in Unsupervised Gaussian Model-Based Clustering} 
In Gaussian model-based clustering, outlier methods usually fall into one of three paradigms: outlier inclusion, \textit{post hoc} outlier identification, and outlier trimming. The first method, proposed by \cite{banfield93}, includes outliers in an additional uniform component over the convex hull. If outliers can be treated as cluster-specific and are symmetric about the component means, we can incorporate them into the tails if we cluster using mixtures of multivariate power exponential \citep{dang15}, Pearson type VII \citep{sun10}, multivariate leptokurtic-normal \citep{bagnato17}, multivariate shifted exponential normal, multivariate tail-inflated normal \citep{tomarchio22}, or multivariate t-distributions \citep{peel00,andrews11a}. \cite{punzo16b} introduce mixtures of contaminated Gaussian distributions, where each cluster has a proportion $\alpha_g\in(0.5,1)$ of `good' points, with density $\phi(\vecx\mid\vecmu_g,\vecSigma_g)$, and a proportion $1-\alpha_g$ of `bad' points, with density $\phi(\vecx\mid\vecmu_g,\eta_g\vecSigma_g)$. Within component~$g$, each distribution has the same mean $\vecmu_g$, but the `bad' points have an inflated covariance matrix $\eta_g\vecSigma_g$, where $\eta_g >1$.

Instead of fitting outliers in the model, we may wish just to identify outliers after the model has been fit. This is of particular interest in novelty detection. \cite{evans15} identify outliers first by fitting a Gaussian mixture model, then quantifying the extent the variance changes when a point is removed.  Points that inflate the variance beyond a certain threshold if included are deemed outliers.

Finally, it may be of interest to trim outliers from the dataset. \cite{cuesta97} developed an impartial trimming approach for $k$-means clustering; however, this method maintains the drawback of $k$-means clustering, where the clusters are spherical with equal --- or, in practice, similar --- radii or they are so well separated that a departure from this shape constraint will not matter. Of course, the latter is a trivial case and is mentioned only for completeness. \cite{garcia08} improved upon trimmed $k$-means with the TCLUST algorithm. TCLUST places a restriction on the eigenvalue ratio of the covariance matrix, as well as implementing a weight on the clusters, allowing for clusters of various elliptical shapes and sizes. \cite{garcia11} extends TCLUST to Gaussian mixture modelling. An obvious challenge with these methods is that the eigenvalue ratio must also be known \textit{a~priori}. There exists an estimation scheme for the proportion of outliers but it is heavily influenced by the choices for number of clusters and eigenvalue ratio. It is of great interest to develop a trimming approach within the model-based clustering domain that does not require pre-specification of the proportion of outliers. This is important because, for most real datasets, the proportion of outliers is \textit{a~priori} unknown.

\subsection{General Unsupervised Outlier Detection Methods}
Some traditional approaches to outlier analysis identify points with outlier scores greater than a certain threshold. For example, we can generate a score for each data point using a measure of central tendency $a$ and a measure of variability $b$. Points are considered outliers if their scores are larger than some threshold $c$, i.e., $\abs{\frac{x-a}{b}}>c$ or $x \notin (a-bc,a+bc)$. Examples include the three sigma rule ($c=3$) and the median absolute difference thresholding strategy \citep{hampel74}.  \cite{yang19} make this method robust by applying the thresholding procedure twice, while \cite{buzzi11} remove outliers one-by-one by iteratively applying a clever mean threshold. 

In the clustering framework, similar outlier scores can be generated using the distance to the cluster centre \citep{hautamaki05} or to the nearest neighbours \citep{ramaswamy00}. We can also consider points of low density to be outliers, for example in \cite{knorr98}, where points with few neighbours within a specified distance are identified, or DBSCAN, which clusters data based on density, with points not placed in clusters classified as noise. \cite{franti18a} smooth noisy data by replacing each point by the mean of its nearest neighbours. They repeat this process several times to obtain better centroid locations and $k$-means classification.  \cite{yang21} take this one step further by treating the distance each point is shifted as an outlier score. 

\cite{boukerche20} provide a detailed survey of methods to address outliers in unsupervised classification. They include approaches which are proximity- or projection-based, and techniques for high-dimensional, streaming, or `big’ data. While these methods could be used to identify observations that deviate from the entire dataset, some cannot consider points which are spurious in relation to the cluster structure. Although \cite{boukerche20} include cluster-based approaches, the methods they consider typically concern groups of outliers which form their own cluster.

\section{Methodology}
\subsection{Distribution of Subset Log-Likelihoods}\label{sec:dist}
In this section, the distribution of squared Mahalanobis distance is used to derive the distribution of subset log-likelihoods, which forms the basis for the proposed OCLUST algorithm. We consider a subset log-likelihood to be  the log-likelihood of a model fitted with $n-1$ of the data points. Formally, if we denote our complete dataset as $\mathcal{X}=\{\vecx_1,\dots,\vecx_n\}$, then we can define the $j${th} subset as the complete dataset with the $j$th point removed, $\mathcal{X}\setminus \vecx_j=\{\vecx_1,\dots,\vecx_{j-1},\vecx_{j+1},\dots, \vecx_n\}$. There are $n$ such subsets. 

Consider a dataset $\mathcal{X}=\{\vecx_1,\dots,\vecx_n\}$ in $p$-dimensional Euclidian space $\mathbb{R}^p$, where each  $\vecx_i \in \mathcal{X}$ has Gaussian mixture model density $f(\vecx_i \mid \vecvartheta)$ as in \eqref{eq:generalmixture1}. The log-likelihood of i.i.d. dataset $\mathcal{X}$ under the Gaussian mixture model is
\begin{equation}
	\ell_{\mathcal{X}}(\vecvartheta)=\sum_{i=1}^n \log \left[\sum_{g=1}^G \pi_g \phi(\vecx_i\mid \vecmu_g, \vecSigma_g)\right],
	\label{eq:loglik}
\end{equation}
which we will denote as $\ell_{\mathcal{X}}$ for simplicity. From \eqref{eq:loglik}, the log-likelihood of the model can be expressed as the sum of the logarithm of each point's probability density. Points near the tails --- corresponding to probability density values close to zero --- will have a smaller contribution to the log-likelihood than points with higher probability density such as points near the centre of the distribution. Intuitively, the larger the difference in log-likelihood between the subset and full models, the more outlying the point is, i.e., the model improved the most in its absence.

By this logic, we treat the point whose absence produced the largest subset log-likelihood, $\vecx_k$ say, as our candidate outlier. Formally, we can express this as follows.
\begin{definition}[Candidate Outlier] We define our candidate outlier as $\vecx_k$, where
	\begin{equation*}
		k=\argmax_{j \in [1,n]} \ell_{\mathcal{X} \setminus \vecx_j}, 
		\label{eq:argmax}
	\end{equation*}
	and $\ell_{\mathcal{X} \setminus \vecx_j}$ is the log-likelihood of the subset with the point $\vecx_j$ removed.
	\label{def:out}
\end{definition}
We continue removing candidate outliers until we obtain our best model, which is determined by the distribution of our subset log-likelihoods, derived below. 
Note that a random variable $W$ from a beta distribution has density 
\begin{equation}\label{eq:beta}
	f_{\text{beta}}(w | \alpha, \beta)=\frac{\Gamma(\alpha+\beta)}{\Gamma(\alpha)\Gamma(\beta)} w^{\alpha-1}(1-w)^{\beta-1},
\end{equation}
where  $0<w<1, \alpha>0, \beta>0$.

The derivation requires the complete-data log-likelihood as an approximation of the log-likelihood. Herein, the complete-data comprise the data $\vecx_1,\dots,\vecx_n$ and their cluster memberships $\vecz_1,\ldots,\vecz_n$, where $\vecz_i=(z_{i1},\ldots,z_{iG})'$, $z_{ig}=1$ if $\vecx_i$ belongs to the $g$th cluster, and $z_{ig}=0$ otherwise.
 In \eqref{eq:loglik}, each cluster contributes to the density of every point. Instead, consider the case where the clusters are well-separated. We can use the complete-data log-likelihood, herein denoted $l_{\mathcal{X}}$, which considers only the density of the component to which the point belongs. We formalize this result in Lemma~\ref{lem:loglik}, but first we must make the following assumption.
\begin{assumption}
The clusters are non-overlapping and well separated.
\label{ass:sep}
\end{assumption}
\noindent In practice, Assumption~\ref{ass:sep} may be relaxed. For more information on the effect of cluster separation, see~\appref{app:overlap}.

\begin{lemma}\label{lem:loglik}
	As the separation between the clusters increases, $\ell_\mathcal{X}\rightarrow l_\mathcal{X}$. In other words, the log-likelihood in \eqref{eq:loglik} converges to $l_\mathcal{X}$, where
	\begin{equation}
		l_{\mathcal{X}}=\sum_{i=1}^n \sum_{g=1}^G z_{ig}\left[\log\pi_g + \log\phi(\vecx_i\mid \vecmu_g, \vecSigma_g)\right].%=\sum_{g=1}^G\sum_{\vecx_i \in \mathcal{C}_g} \left[\log\pi_g + \log\phi(\vecx_i\mid \vecmu_g, \vecSigma_g)\right].
		\label{eq:logass}
	\end{equation}
\end{lemma}

\noindent A proof of Lemma~\ref{lem:loglik} may be found in~\appref{app:proofs}. We will maintain Assumption \ref{ass:sep} throughout this paper. The quantity $l_\mathcal{X}$ denotes the complete-data log-likelihood for the entire dataset~$\mathcal{X}$, and we use $l_{\mathcal{X}\setminus \vecx_j}$ to denote the complete-data log-likelihood for the $j${th} subset $\mathcal{X}\setminus \vecx_j$. Finally,  we define the variable $Y_j=l_{\mathcal{X} \setminus \vecx_j}-l_\mathcal{X}$ as the difference between the $j${th} subset complete-data log-likelihood and the complete-data log-likelihood for the entire dataset. This leads to our main result, which follows as Theorem~\ref{prop2}.
\begin{theorem}\label{prop2}
	For a point $\vecx_j$ belonging to the $h$th cluster, i.e., $z_{jh}=1$, if $l_\mathcal{X}$ is the complete-data log-likelihood and $Y_j=l_{\mathcal{X} \setminus \vecx_j}-l_\mathcal{X}$, then $Y_j \mid (z_{jh}=1) $ has an approximate shifted and scaled beta density, i.e.,
	\begin{equation}
		\label{eq:L}
		Y_j \mid (z_{jh}=1) \sim f_{\text{beta}}\left(\frac{2n_h}{(n_h-1)^2} (y_j-c)~\bigg|~\frac{p}{2},\frac{n_h-p-1}{2}\right)
	\end{equation}
	for $c<y_j<\frac{(n_h-1)^2}{2n_h}+c, n_h>p+1$, where $c=-\log\hat{\pi}_h+\frac{p}{2}\log(2\pi)+\frac{1}{2}\log\abss{\sampcov_h}$, {$n_h$ is the number of points in cluster $h$}, $\hat{\pi}_h=n_h/n$, $$\sampcov_h=\frac{1}{n_h-1}\sum_{i=1}^n z_{ih}(\vecx_i-\bar{\vecx}_h)(\vecx_i-\bar{\vecx}_h)'$$ is the sample covariance matrix of cluster $h$, and $\bar{\vecx}_h=\frac{1}{n_h}\sum_{i=1}^n z_{ih}\vecx_i.$
\end{theorem}
\begin{proof}
Our parameters are unknown in unsupervised classification, so we must estimate them. Due to their desirable qualities, we will replace $\pi_g$, $\vecmu_g$ and $\vecSigma_g$ by their unbiased estimates:
\begin{equation}\label{eqn:samp_est}
\hat{\pi}_g=n_g/n,\qquad
		\hat{\vecmu}_g=\bar{\vecx}_g,\qquad
		\hat{\vecSigma}_g=\sampcov_g,
\end{equation}
where $n_g=\sum_{i=1}^n z_{ig}$ is the number of observations in the $g$th cluster, $\hat{\pi}_{g}$ is the sample proportion, $\bar{\vecx}_g$ is the sample mean, and $\sampcov_g$ is the sample covariance for the $g${th} cluster considering all observations in the entire dataset $\mathcal{X}$. Now consider $\hat{\pi}_{g\setminus j}$, $\bar{\vecx}_{g \setminus j}$ and $\sampcov_{g\setminus j}$, the sample proportion, sample mean, and sample covariance matrix, respectively, for the $g${th} cluster considering only observations in the $j${th} subset $\mathcal{X}\setminus \vecx_j$. We will require $\hat{\pi}_{g\setminus j}\approx\hat{\pi}_g$, $\bar{\vecx}_{g\setminus j}\approx \bar{\vecx}_g$  and $\sampcov_{g\setminus j}\approx\sampcov_g$ for all $j$. Because we have $\hat{\pi}_{g\setminus j}\rightarrow\hat{\pi}_g$ $\bar{\vecx}_{g\setminus j}\rightarrow \bar{\vecx}_g$  and $\sampcov_{g\setminus j}\rightarrow\sampcov_g$ as $n_g\rightarrow\infty$ (proof in \appref{app:estconv}), we need the following assumption.
	\begin{assumption}
		The number of observations in each cluster, $n_g$, is large.
	\end{assumption}

		When $n_g$ is large, sample parameter estimates $\hat{\pi}_g$, $\bar{\vecx}_g$, and $\sampcov_g$, $g \in [1,G]$, approach the true parameters and they are treated as constant for each subset $\mathcal{X} \setminus \vecx_j$, $j \in [1,n]$. Thus, the complete-data log-likelihood for the $j$th subset, $\mathcal{X} \setminus \vecx_j$, when $z_{jh}=1$ is
		\begin{equation}\label{eqn:7}
			l_{\mathcal{X} \setminus \vecx_j} \approx l_\mathcal{X} - \log\hat{\pi}_h- \log\phi(\vecx_j\mid\bar{\vecx}_h,\sampcov_h).
		\end{equation}
		Rearranging \eqref{eqn:7} yields
		\begin{equation}
			\begin{split}
			l_{\mathcal{X} \setminus \vecx_j}-l_\mathcal{X}
				% &= -\log\pi_g- \log\phi(\vecx_j\mid\vecmu_g,\vecSigma_g)\\
				%&=-\log\pi_g+\frac{p}{2}\log(2\pi)+\frac{1}{2}\log\abss{\vecSigma_g}+\frac{1}{2}(\vecx_j-\vecmu_g)'\vecSigma_g^{-1}(\vecx_j-\vecmu_g)\\
				&\approx-\log\hat{\pi}_h+\frac{p}{2}\log(2\pi)+\frac{1}{2}\log\abss{\sampcov_h}+\frac{1}{2}t_j,
			\end{split}
			\label{eq:difflikpop}
		\end{equation}
		where $t_j=(\vecx_j-\bar{\vecx}_h)'\sampcov_h^{-1}(\vecx_j-\bar{\vecx}_h)$ is the squared Mahalanobis distance for point $\vecx_j$ using the sample parameter estimates in \eqref{eqn:samp_est} when $z_{jh}=1$. 
	%MD lemma%
	
	\begin{lemma}
		Sample squared Mahalanobis distance is distributed according to a scaled beta distribution \citep[]{gnanadesikan72}. When the data are multivariate normally distributed, i.e., $\vecX \sim \text{MVN}\left(\vecmu,\vecSigma\right)$, $$\frac{n}{(n-1)^2}T_j \sim f_{\text{beta}}\left(\frac{n}{(n-1)^2} t_j~\Bigg|~\frac{p}{2},\frac{n-p-1}{2}\right)$$ for $0\leq t_j \leq (n-1)^2/n$.
		\label{lem:tbeta}
	\end{lemma}
\noindent \cite{ververidis08} prove Lemma~\ref{lem:tbeta} for all $n, p$ satisfying $p<n<\infty$. 

Finally, we will perform a change of variables to prove our main result from Theorem~\ref{prop2}. Let $W_j=\frac{n_h}{(n_h-1)^2}T_j$ and $Y_j\mid (z_{jh}=1)=\frac{1}{2}T_j+c$, where $c=-\log\hat{\pi}_h+\frac{p}{2}\log(2\pi)+\frac{1}{2}\log\abss{\sampcov_h}$. Then 
	\begin{equation*}
		Y_j\mid (z_{jh}=1)=\frac{(n_h-1)^2}{2n_h}W_j+c.
	\end{equation*}
	Because $W_j$ is beta distributed, it has density of the from \eqref{eq:beta}. The change of variables allows the density of $Y_j\mid (z_{jh}=1)$ to be written
	\begin{equation}
		f_{Y\mid (z_{jh}=1)}(y_j)=\frac{2n_h}{(n_h-1)^2}\frac{\Gamma(\alpha+\beta)}{\Gamma(\alpha)\Gamma(\beta)}\left[\frac{2n_h}{(n_h-1)^2}(y_j-c)\right]^{\alpha-1}\left[1-\frac{2n_h}{(n_h-1)^2}(y_j-c)\right]^{\beta-1}, \\
		\label{eq:ydens}
	\end{equation}
	for $c<y_j<\frac{(n_h-1)^2}{2n_h}+c, \alpha>0, \beta>0$. Thus, $Y_j\mid (z_{jh}=1)$ is distributed according to a shifted and scaled beta distribution, i.e.,
	\begin{equation}
		Y_j \mid (z_{jh}=1) \sim f_{\text{beta}}\left(\frac{2n_h}{(n_h-1)^2} (y_j-c)~\bigg|~\frac{p}{2},\frac{n_h-p-1}{2}\right).
	\end{equation}
	
	\end{proof}
	%\begin{rmk}
	Now, the density for $Y_j$ is conditional on $z_{jh}$. Denote this density by $f_h(y)$. Because $P(z_{jh}=1)=\pi_h$ for all $j$, the density of $Y$, unconditional on $\vecz$ is
	\begin{equation}
		f(y\mid\vecvartheta)=\sum_{g=1}^{G}{\pi}_g  f_g(y \mid \vectheta_g),
		\label{eq:mixdens}
	\end{equation}
	where $f_g(y \mid \vectheta_g)$ is a shifted and scaled beta density described in \eqref{eq:ydens}, and $\vectheta_g=\{n_g,p,\hat{\pi}_g, \bar\vecx_g, \sampcov_g\}$.
	%\end{rmk}
	\begin{rmk}\label{rem:ass}
		It is important to note that \eqref{eq:mixdens} does not depend on the choice of $\vecx_j$. Rather, it is the distribution for $n-1$ of the points with any one point $\vecx_j$ removed. The random variable $Y$ has density $f(y \mid \vecvartheta)$ from \eqref{eq:mixdens} when the data arise from a finite Gaussian mixture model without outliers. Divergence from this distribution indicates a misspecified model, which we assume to be due to the presence of outliers. This forms the basis for the OCLUST algorithm.\end{rmk}
	
	\subsection{OCLUST Algorithm}\label{sec:algorithms}
	Consider $y_1,\ldots,y_n$ to be realizations of the random variable $Y$, which has density given by \eqref{eq:mixdens}. We propose testing the adherence of $y_1,\ldots,y_n$ to the reference distribution in \eqref{eq:mixdens} as a way to test for the presence of outliers. In other words, if $y_1,\ldots,y_n$ does not have a beta mixture distribution, then we assume outliers are present in the model.  Because ${\ell}_\mathcal{X}$ converges to $l_\mathcal{X}$ and acts as a good approximation, we will use $\ell_\mathcal{X}$. This is important because we will need $\ell_{\mathcal{X} \setminus \vecx_j}$ for outlier identification and, additionally, it is outputted by many existing clustering algorithms. The algorithm described below uses the log-likelihood and parameter estimates calculated using the expectation-maximization (EM) algorithm \citep{dempster77} for Gaussian model-based clustering; however, other methods may be used for parameter estimation.
	
	The proposed algorithm is called OCLUST. It both identifies likely outliers and estimates the proportion of outliers within the dataset. When $\vecX$ is distributed according to a Gaussian mixture model, then ${Y}$ has the probability density in \eqref{eq:mixdens}. If $y_1,\ldots,y_n$ does not come from \eqref{eq:mixdens}, then $\vecx_1,\ldots,\vecx_n$ does not arise from a Gaussian mixture model \eqref{eq:generalmixture1}, implying that the model is misspecified. We assume the misspecification is due to outliers.
	The OCLUST algorithm assumes that the model is otherwise correctly specified and when $Y$ does not follow the distribution in \eqref{eq:mixdens}, then outliers are present. The `closeness' of the distribution of $Y$ to \eqref{eq:mixdens} is assessed using the Kullback-Leibler (KL) divergence, estimated via relative frequencies.  The algorithm (Algorithm~\ref{OCLUST}) involves removing candidate outliers one-by-one until KL divergence is minimized. At each step, we obtain the log-likelihood of the entire dataset and the log-likelihoods of the $n$ subsets. We then assess the distribution of subset log-likelihoods and  determine our next candidate outlier, $t$, according to Definition~\ref{def:out}. This candidate outlier is then trimmed and the algorithm continues to the next iteration. Notably, KL divergence generally decreases as outliers are removed and the model improves. Once all outliers are removed, KL divergence increases again as points are removed from the tails.  We select the number of outliers as the location of the minimum KL divergence.
	
	The OCLUST algorithm is outlined in Algorithm~\ref{OCLUST}, with $n$ being the size of the dataset, $G$ the number of clusters, $B \in [0,F-1]$ the number of initially rejected outliers (see Remark~\ref{rmk:B}), and $F$ the chosen upper bound (see Remark~\ref{rmk:F}). An \textsf{R} implementation of the algorithm is available as the \texttt{oclust} package on CRAN \citep{clark19}. 
	
	\begin{rmk}[Initialization]
		In practice, we do not need to start the algorithm with the entire dataset. We may first remove obvious, gross outliers and start the algorithm from this intermediate solution. This greatly reduces computation time by reducing the number of iterations and it may improve classification by decreasing the likelihood of gross outliers being placed into their own clusters. \label{rmk:B}
	\end{rmk}
	
	\begin{rmk}[Choice of $F$]
		We choose $F$ as an upper bound for the number of outliers. It is used to limit the number of iterations for the algorithm to reduce the computation time. $F$ may be chosen using prior knowledge about the type of data, or we may choose large $F$ to be conservative. In the absence of any information about the dataset, we could choose, e.g., $F=n/4$, setting the maximum proportion of outliers to be 25\%.\label{rmk:F}
	\end{rmk}	
\begin{algorithm}[!ht]
		\caption{OCLUST algorithm}\label{OCLUST}
		\begin{algorithmic}[1]
			\Procedure{OCLUST}{$\mathcal{X},n,G,F$}
			\State (Optional) Identify gross outliers, $b$, using method of choice. $B=\#b$. 
			\State \begin{varwidth}[t]{\linewidth}
				Update:\par
				\hskip\algorithmicindent $n \hookleftarrow n-B$\par
				\hskip\algorithmicindent 	$\mathcal{X} \hookleftarrow \mathcal{X} \setminus b$
			\end{varwidth}\newline
			\For{$f $ in $B:F$}
			\State \begin{varwidth}[t]{\linewidth}
				Cluster the data $\mathcal{X}$ into $G$ clusters, using a Gaussian model-based \newline
				\hskip\algorithmicindent clustering algorithm, e.g., the EM algorithm in the \texttt{mixture} package \label{step: cluster}	\end{varwidth} \newline
			\State Output $\ell_\mathcal{X}$, and for each cluster save: $\sampcov_g$, $n_g$,  $\hat\pi_g$. \label{step:paras}
			\For{$j $ in $1:n$}
			\State \begin{varwidth}[t]{\linewidth}
				Cluster the subset $\mathcal{X} \setminus \vecx_j$ into $G$ clusters, using the method chosen\newline
				\hskip\algorithmicindent in Step \ref{step: cluster}.	\end{varwidth} \newline
			\State Output $\ell_{\mathcal{X} \setminus \vecx_j}$ and calculate $y_j=\ell_{\mathcal{X} \setminus \vecx_j}-\ell_\mathcal{X}$.
			\EndFor
			\State \label{step:dens}Generate the density of $Y$ using \eqref{eq:mixdens} and the parameters from Step~\ref{step:paras}.
			\State \begin{varwidth}[t]{\linewidth}
				Calculate the approximate KL divergence of $y_1,\ldots,y_n$   to the density\newline
				\hskip\algorithmicindent in Step~\ref{step:dens}, using relative frequencies.	\end{varwidth} \newline
			
			\State Determine the most likely outlier $t$ as per Definition~\ref{def:out}.
			\State \begin{varwidth}[t]{\linewidth}
				Update:\par
				\hskip\algorithmicindent $n \hookleftarrow n-1$\par
				\hskip\algorithmicindent 	$\mathcal{X} \hookleftarrow \mathcal{X} \setminus t$
			\end{varwidth}\newline
			\EndFor
			\State Choose $f$ for which the KL divergence is minimized \newline 
			\Comment{This is the predicted number of outliers.} \newline
			\Comment{Use the model corresponding to iteration $f$}
			\EndProcedure
		\end{algorithmic}
	\end{algorithm}
	%\clearpage
	
	\subsection{Computational Complexity}
	Clustering must be performed $n~+~1$ times for each of the $F+1$ iterations of the algorithm.  In our simulation study, we use the \texttt{mixture} package to implement the EM algorithm, the order of which is generally $O(Gp^3n)$. As a result, OCLUST is $O(FGp^3n^2)$. The computation time can be reduced by initializing the subset models with the parameters from the full model to vastly decrease the number of iterations required within the clustering algorithm. In addition, subset models are independent and can easily be computed in parallel. 
	Computation times for applying OCLUST to the datasets in Section~\ref{sec:apps} are available in \appref{app:timing}. These times are similar to or somewhat faster than the mean-shift outlier detection algorithm when OCLUST is run in parallel. It is important to note that these times for OCLUST take into account the algorithm running to the upper bound $F$, which is often unnecessary. Thus, we propose an early stopping criterion below. 	
	
	The original version of the OCLUST algorithm involves calculating the KL divergence over a range of possible numbers of outliers. The best model is then chosen as the one which minimizes the KL divergence. Instead, consider a stopping rule for the algorithm to minimize the number of iterations required. We can stop the algorithm once the data fit well to the distribution. One way to do this is using Kuiper's Test \citep{kuiper60}, where 
	\begin{align*}
		\text{H}_0&:\text{The data arise from the specified distribution.}\\
		\text{H}_A&:\text{The data contradict the specified distribution.}
	\end{align*}
	By conducting Kuiper's test, we compare the empirical cumulative distribution function (CDF) to the CDF of the proposed beta mixture distribution. The test statistic is of the form
$$T_0=D^+ +D^-,$$
where $D^+=\max[F_E(y_i)-F_0(y_i)]$ and $D^-=\max[F_0(y_i)-F_E(y_i)]$, for $i\in [1,n]$, $F_0(y_i)$ is the CDF of $Y$ at point $y_i$, and $F_E(y)$ is the empirical CDF at point $y_i$ for a sample of size $n$, i.e., $F_E(y_i) =i/n$.
	An approximate p-value can be estimated using Monte Carlo simulation from the proposed distribution in \eqref{eq:mixdens}. A total of $m$ datasets of size $n$ are simulated and the test statistic is calculated for each one. Then, we estimate the p-value as $(r+1)/(m+1)$, where $m$ is the number of samples generated and $r$ is the number of samples to produce a test statistic greater than or equal to $T_0$ \citep{north02}. Instead of using KL divergence in the OCLUST algorithm, we calculate the approximate p-value for each iteration and stop once it is higher than a pre-specified significance level, e.g., 5\% or 10\%.

	\section{Applications}\label{sec:apps}
	%This section will evaluate the performance of OCLUST against other outlier detection methods using simulated and real datasets. 
	\subsection{Illustrative Example}\label{sec:example}
To illustrate how the OCLUST algorithm works, we return to the dataset in Figure~\ref{fig:toy}, but add three uniform outliers instead of 100. We run the OCLUST algorithm using the \texttt{oclust} package \citep{clark19}, setting $B=0$ and $F=10$. We achieve minimum KL divergence after the three outliers are removed. The optimal solution and corresponding KL graph are plotted in Figure~\ref{fig:illex}.
\begin{figure*}[!ht]
	\centering
	\begin{subfigure}[t]{3.2in}
		\centering
		\includegraphics[width=3.2in]{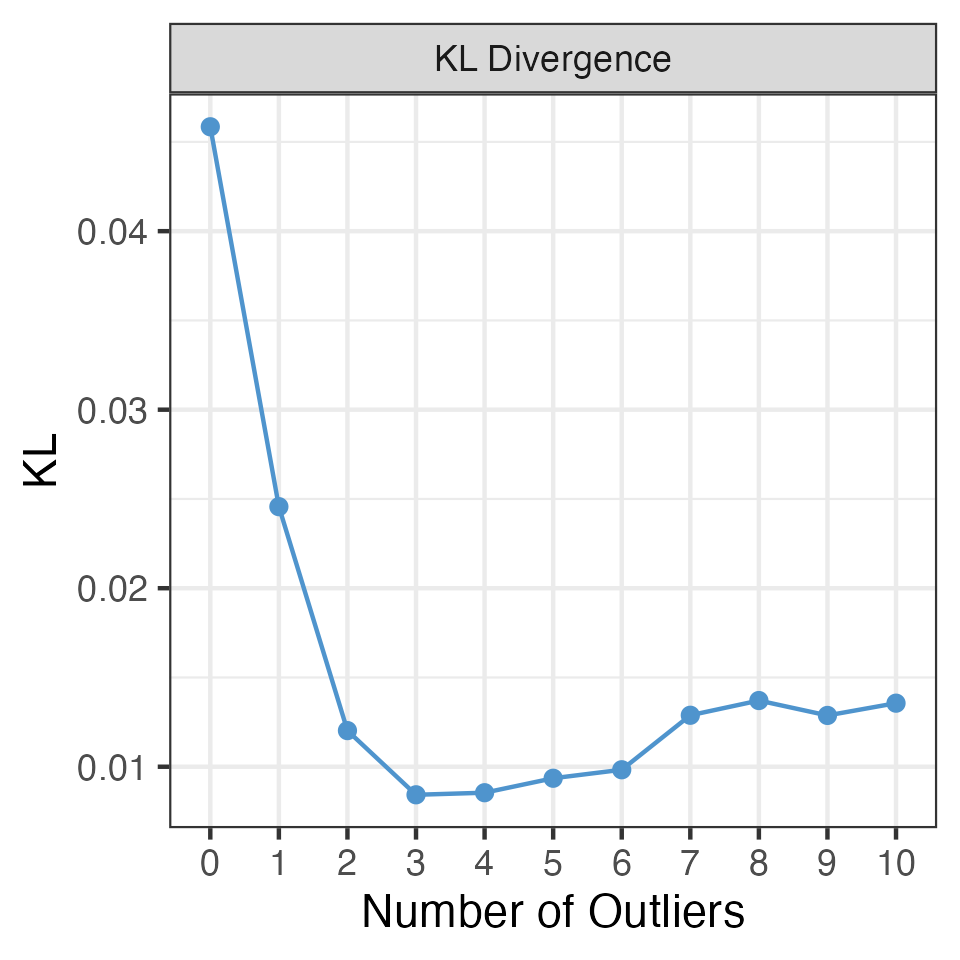}
	\end{subfigure}
	~
\hspace{-0.25in}	\begin{subfigure}[t]{3.2in}
		\centering
		\includegraphics[width=3.2in]{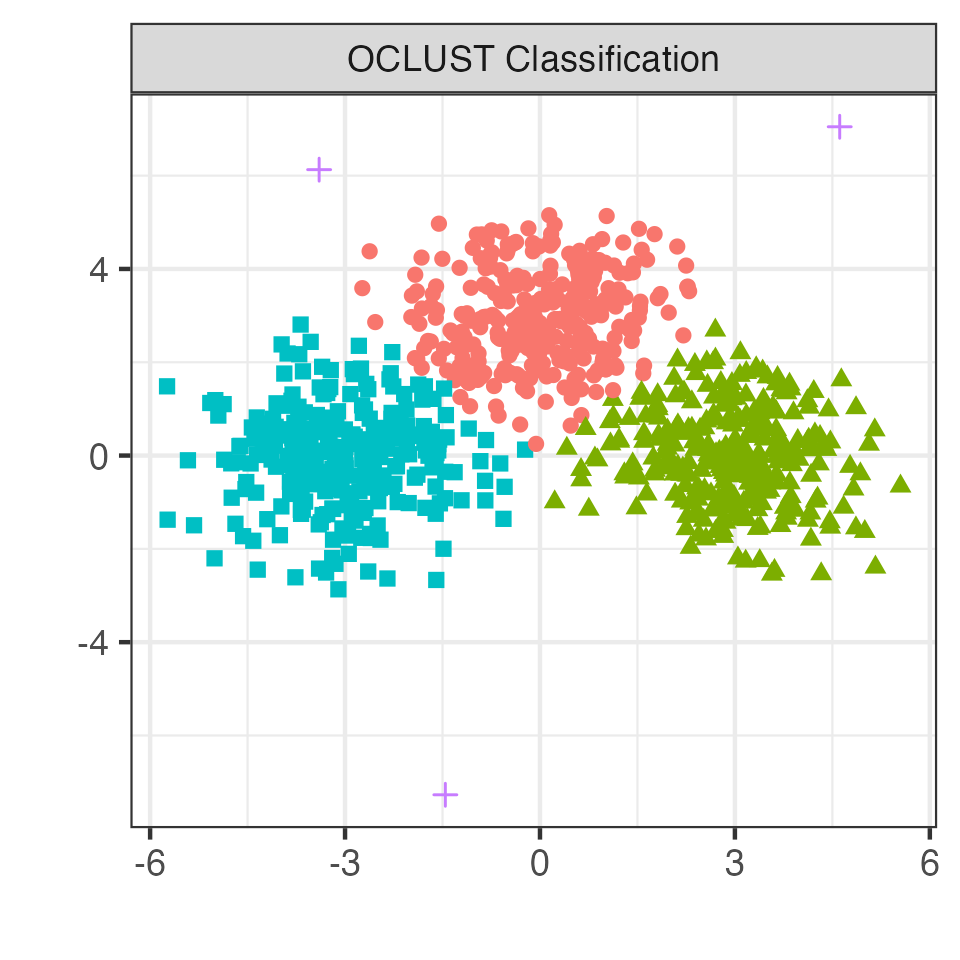}
	\end{subfigure}%
	\caption{Results from the illustrative example in Section~\ref{sec:example}, with a graph of KL divergence (left) and predicted classifications by the OCLUST algorithm (right). }
\label{fig:illex}
\end{figure*}
	
\subsection{Simulation Study}\label{sec:sims}
	The following simulation study tests the performance of OCLUST against the following five popular algorithms:
	\begin{enumerate}[label=\alph*.]
		\item Contaminated normal mixtures \citep[CNMix;][]{punzo16b};
		\item Noise component mixtures --- mixtures of Gaussian clusters and a uniform component \citep[NCM;][]{banfield93}; 
		\item DBSCAN; 
		\item Mean-shift outlier detection and filtering \citep{franti18a}; and
		\item 2T \citep{yang19}, i.e., a thresholding approach.
	\end{enumerate}
	
	OCLUST is a trimming method which removes outliers from the model. In contrast, CNMix and NCM  are outlier inclusion methods. CNMix treats outliers as contamination and assumes the outliers are cluster-specific and reside around the clusters. NCM assumes that outliers are included uniformly. DBSCAN treats clusters as density-connected points and outliers as points that do not belong to clusters. Mean-shift outlier detection is a pre-processing step for a subsequent clustering algorithm. Finally, the thresholding approach considers a point an outlier if the outlier `score' exceeds a specified threshold. Following \cite{yang19}, we apply this procedure twice (2T). Calculated before clustering, we choose the outlier `score' to be the squared Mahalanobis distance from the centre of the dataset. This measure accounts for differences in variability among dimensions. Subsequently, we cluster the remaining data using the EM algorithm in the \texttt{mixture} package. These methods are broadly classified into two types: those that remove outliers and then cluster (mean-shift and {2T}) and those that do both simultaneously (OCLUST, CNMix, NCM, DBSCAN). 
	
	The following simulation scheme closely follows that in \cite{franti18a}. We use the same eight two-dimensional benchmark datasets  plus one 32-dimensional dataset \citep{franti18b}. Sets S1--S4 investigate increasing degrees of overlap. There is a mix of 15 spherical and non-spherical Gaussian clusters with overlap ranging from 9\% to 44\%.  Sets A1--A3 have increasing numbers of clusters, with all sets having spherical Gaussian clusters with equal cluster sizes, deviation, and overlap. The Unbalance set has three clusters densely populated with 2000 points each and five clusters sparsely populated with 100 points each. The Dim032 dataset has 16 well-separated Gaussian clusters in 32 dimensions. The datasets are summarized in Table~\ref{tab:data} and plotted in Figure~\ref{fig:data}, with a two-dimensional projection shown for the Dim032 dataset. 	
	\begin{table}[!ht]
		\centering
		\caption{Datasets used in the simulation study, where U and D32 represent the Unbalance and Dim032 sets, respectively, and $F$ indicates the upper bound used in the OCLUST algorithm.}	 \label{tab:data}
		\begin{tabular}{r|ccccccccc}
			\hline
			Dataset & A1& A2& A3& S1 & S2 & S3& S4& U & D32\\
			\hline
			Size & 3000& 5250& 7500& 5000 & 5000 & 5000& 5000& 6500 & 1024\\
			Clusters & 20& 35& 50& 15 & 15 & 15& 15& 8 & 16\\
			Noise & 210& 368& 525& 350 & 350 & 350& 350& 455 & 72\\
			\hline
			$F$ &300&525&750&500&500&500&500&650&110\\
			\hline
		\end{tabular} 
	\end{table}
	\begin{figure}[!htb]
		\centering
		\includegraphics[width=\textwidth]{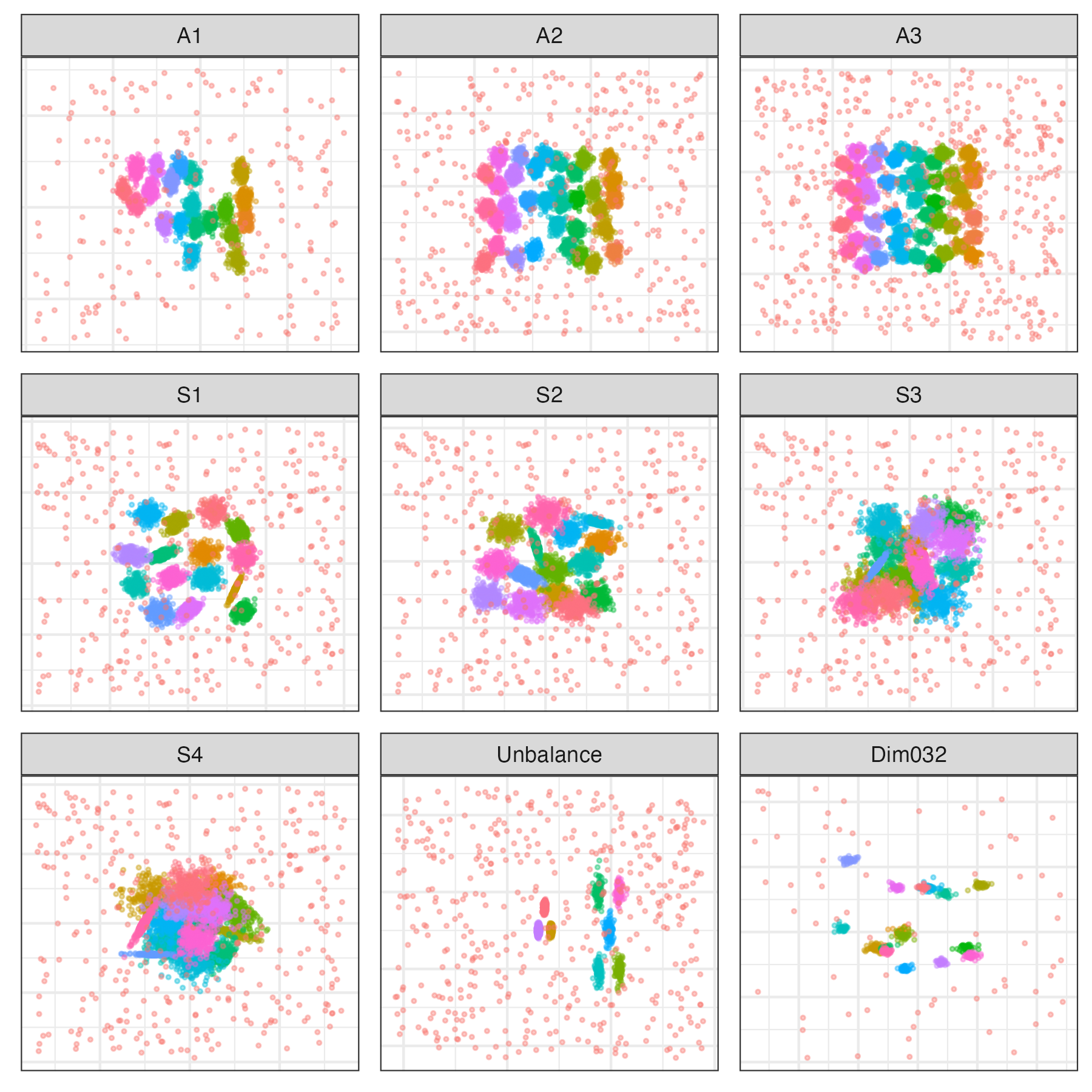}
		\caption{The nine benchmark datasets used herein, where red points represent added uniform noise and other colours correspond to the true classes.}
		\label{fig:data}
	\end{figure}
	
	Following \cite{franti18a}, we set the proportion of outliers at 7\%. We generate uniform noise in each dimension in the range $[x_{\text{mean}}-2\times\text{range},x_{\text{mean}}+2\times\text{range}]$, where the range is the difference between the farthest observation and the mean in each dimension.
	
	We run {\tt oclust} in \textsf{R}, fixing the upper bound $F$ from Table~\ref{tab:data}, i.e., at or about 10\% of the original data size. We eliminate initial gross outliers with large squared Mahalanobis distances to regions of high density identified by the \texttt{dbscan} package \citep{hahsler19}. Covariance matrices for the S sets, A sets, and Unbalance set are unrestricted, allowing varying volumes, shapes, and orientations among the clusters. We restrict the Dim032 set to a variance structure with varying volumes and spherical shapes to hasten computation time.  We run CNMix using the {\tt CNmixt} function from the {\tt ContaminatedMixt} package \citep{punzo18} with default $k$-means initialization. The {\tt CNmixt} function chooses the best model for variance and contamination using the BIC. We run NCM using the {\tt Mclust} function in the {\tt mclust package} \citep{scrucca16}, initializing the noise component as a random sample of points with probability 1/4. We let the package choose the best model using the Bayesian information criterion \citep{schwarz78}. We run mean-shift outlier detection in Python using \cite{yangcode}'s code. We specify the number of clusters and perform three iterations. Due to its superior performance in \cite{franti18a},  non-outlying points are clustered using random swap clustering in Python \citep{yang21,franti18} and run for the recommended 5000 iterations.  We implement the DBSCAN algorithm with the \texttt{dbscan} package, setting the minimum number of points per cluster to be 20. For the thresholding approach, we calculate the squared Mahalanobis distance from the centre of the dataset and trim those that are more than three standard deviations away from the mean. This is performed twice, and the remaining points are clustered with the \texttt{mixture} package. 
	
	The performance of each method can be assessed using three approaches for evaluating classification accuracy --- adjusted Rand index \citep[ARI;][]{hubert85}, normalized mutual information \citep[NMI;][]{kvalseth87}, and centroid index \citep[CI;][]{franti14} --- and two approaches of outlier detection performance --- true positive rate (TPR) and false positive rate (FPR). The ARI is a measure of agreement between two partitions, e.g., between clusters and true classes. The ARI equals 1 with  perfect agreement and has expected value 0 when randomly partitioned. Negative values indicate partitions with less agreement than would be expected by random chance. NMI is a measure of the amount of information shared between the true classes and the clustering solution. It is a scale from 0 to 1, with 0 corresponding to no mutual information and 1 corresponding to perfect agreement. CI is a measure of dissimilarity among cluster centroids, where $\text{CI}=0$ reflects one-to-one correspondence between the estimated cluster centres and the ground truth. We calculate ARI and NMI for cluster agreement {taking into account whether points are classified into the correct group, considering outliers as one of those groups.} TPR reflects the proportion of true outliers classified as such and FPR reflects the proportion of non-outliers classified as outliers. The results are shown in Table~\ref{tab:simres}, with the best algorithm according to each approach bolded. 
	\begin{table}
		\begin{center}
			\caption{Results of the six outlier identification algorithms on the benchmark datasets. ARI and NMI consider cluster classification with `outlier' being considered a class.} \label{tab:simres}
			\resizebox{0.49\textwidth}{!}{
				\begin{tabular}{lrrrrr}  
					\multicolumn{6}{c}{A1 Dataset} \\
					\hline 
					 Algorithm & ARI & NMI & CI & TPR & FPR \\ 
					\hline
					OCLUST & 0.96 & 0.97 &   \textbf{0} & \textbf{0.87 }& \textbf{0.00} \\ 
					Mean-Shift &\textbf{ 0.97} & \textbf{0.98} &  \textbf{ 0} & 0.82 &  \textbf{0.00} \\ 
					2T & 0.84 & 0.92 &   1 & 0.69 &  \textbf{0.00} \\ 
					CNMix & 0.74 & 0.86 &   5 & 0.00 &  \textbf{0.00}\\ 
					NCM & 0.79 & 0.89 &   3 & 0.80 &  \textbf{0.00} \\ 
					DBSCAN & 0.01 & 0.06 &  19 & 0.85 &  \textbf{0.00}\\ 
					\hline
			\end{tabular}
		}
			\hfil   %<---
			\resizebox{0.49\textwidth}{!}{
				\begin{tabular}{lrrrrr}  
					\multicolumn{6}{c}{A2 Dataset} \\
					\hline 
					Algorithm & ARI & NMI & CI & TPR & FPR \\ 
					\hline
					OCLUST & \textbf{0.95} & \textbf{0.97} &   \textbf{0} & \textbf{0.82} & \textbf{0.00} \\ 
					Mean-Shift & 0.94 & \textbf{0.97} &   \textbf{0} & 0.73 & \textbf{0.00} \\ 
					2T & 0.81 & 0.93 &   4 & 0.68 & \textbf{0.00} \\ 
					CNMix & 0.72 & 0.89 &   7 & 0.00 &\textbf{0.00} \\ 
					NCM & 0.55 & 0.75 &  17 & 0.60 & \textbf{0.00} \\ 
					DBSCAN & 0.00 & 0.04 &  34 & 0.73 & \textbf{0.00} \\ 
					\hline
			\end{tabular}
		}
			\medskip
			\resizebox{0.49\textwidth}{!}{
				\begin{tabular}{lrrrrr}  
					\multicolumn{6}{c}{A3 Dataset} \\
					\hline 
					Algorithm & ARI & NMI & CI & TPR & FPR \\ 
					\hline
					OCLUST & \textbf{0.94 }& \textbf{0.97} &   \textbf{0} & \textbf{0.83 }& \textbf{0.00} \\ 
					Mean-Shift & 0.93 & \textbf{0.97} &   \textbf{0} & 0.72 & \textbf{0.00} \\ 
					2T & 0.84 & 0.95 &   4 & 0.67 & \textbf{0.00} \\ 
					CNMix & 0.72 & 0.90 &  11 & 0.00 & \textbf{0.00} \\ 
					NCM & 0.53 & 0.76 &  26 & 0.56 & \textbf{0.00} \\ 
					DBSCAN & -0.00 & 0.04 &  49 & 0.68 & \textbf{0.00} \\ 
					\hline
			\end{tabular}
		}
			\hfil   %<---
			\resizebox{0.49\textwidth}{!}{
				\begin{tabular}{lrrrrr}  
					\multicolumn{6}{c}{S1 Dataset} \\
					\hline 
					Algorithm & ARI & NMI & CI & TPR & FPR \\ 
					\hline
					OCLUST & \textbf{0.96} & \textbf{0.96} &   \textbf{0} & 0.89 & 0.01 \\ 
					Mean-Shift & 0.95 & \textbf{0.96} &    \textbf{0} & 0.87 & 0.01 \\ 
					2T & 0.95 & \textbf{0.96} &    \textbf{0}  & 0.72 & \textbf{0.00} \\ 
					CNMix & 0.80 & 0.87 &   2 & 0.00 & \textbf{0.00} \\ 
					NCM & 0.90 & 0.93 &   1 & 0.85 & \textbf{0.00} \\ 
					DBSCAN & 0.01 & 0.18 &    \textbf{0}  & \textbf{1.00} & 0.81 \\ 
					\hline
			\end{tabular}
		}
			\medskip
			\resizebox{0.49\textwidth}{!}{
				\begin{tabular}{lrrrrr}  
					\multicolumn{6}{c}{S2 Dataset} \\
					\hline 
					Algorithm & ARI & NMI & CI & TPR & FPR \\ 
					\hline
					OCLUST & 0.91 & 0.92 &   \textbf{0}& 0.81 & \textbf{0.00} \\ 
					Mean-Shift & \textbf{0.92} & \textbf{0.93} &   \textbf{0} & 0.84 & \textbf{0.00} \\ 
					2T & 0.70 & 0.83 &   1 & 0.73 & \textbf{0.00} \\ 
					CNMix & 0.72 & 0.82 &   3 & 0.07 & 0.02 \\ 
					NCM & 0.86 & 0.90 &   1 & 0.82 & \textbf{0.00} \\ 
					DBSCAN & 0.01 & 0.13 &   2 & \textbf{1.00} & 0.87 \\ 
					\hline
			\end{tabular}
		}
			\hfil   %<---
			\resizebox{0.49\textwidth}{!}{
				\begin{tabular}{lrrrrr}  
					\multicolumn{6}{c}{S3 Dataset} \\
					\hline 
					Algorithm & ARI & NMI & CI & TPR & FPR \\ 
					\hline
					OCLUST & \textbf{0.72 }& \textbf{0.79 }&   \textbf{0} & 0.85 & 0.01 \\ 
					Mean-Shift & 0.71 & 0.78 &   \textbf{0} & 0.83 & 0.01 \\ 
					2T & 0.42 & 0.67 &   1 & 0.73 & \textbf{0.00} \\ 
					CNMix & 0.54 & 0.69 &   4 & 0.00 & \textbf{0.00}\\ 
					NCM & 0.55 & 0.72 &   3 & 0.80 & \textbf{0.00} \\ 
					DBSCAN & 0.01 & 0.07 &   9 & \textbf{1.00} & 0.93 \\ 
					\hline
			\end{tabular}
		}
			\medskip
			\resizebox{0.49\textwidth}{!}{
				\begin{tabular}{lrrrrr}  
					\multicolumn{6}{c}{S4 Dataset} \\
					\hline 
					Algorithm & ARI & NMI & CI & TPR & FPR \\ 
					\hline
					OCLUST & 0.42 & 0.65 &   1 & 0.91 & 0.02 \\ 
					Mean-Shift & \textbf{0.63} & \textbf{0.72} &   \textbf{0} & 0.86 & 0.01 \\ 
					2T & 0.16 & 0.50 &   1 & 0.79 & \textbf{0.00} \\ 
					CNMix & 0.40 & 0.59 &   4 & 0.00 & \textbf{0.00} \\ 
					NCM & 0.56 & 0.69 &   2 & 0.83 & \textbf{0.00} \\ 
					DBSCAN & 0.01 & 0.11 &   8 & \textbf{1.00} & 0.89 \\ 
					\hline
			\end{tabular}
		}
			\hfil   %<---
			\resizebox{0.49\textwidth}{!}{
				\begin{tabular}{lrrrrr}  
					\multicolumn{6}{c}{Unbalance Dataset} \\
					\hline 
					Algorithm & ARI & NMI & CI & TPR & FPR \\ 
					\hline
					OCLUST & \textbf{1.00} & \textbf{0.99} &   \textbf{0} & 0.96 & \textbf{0.00} \\ 
					Mean-Shift & 0.43 & 0.58 &   5 & \textbf{1.00 }& 0.10 \\ 
					2T & 0.88 & 0.85 &   3 & 0.80 & 0.01 \\ 
					CNMix & 0.99 & 0.94 &   3 & 0.00 & \textbf{0.00} \\ 
					NCM & 0.99 & 0.96 &   2 & 0.91 & \textbf{0.00} \\ 
					DBSCAN & 0.98 & 0.91 &   \textbf{0} & 0.99 & 0.05 \\ 
					\hline
			\end{tabular}
		}
			\medskip
			\resizebox{0.49\textwidth}{!}{
				\begin{tabular}{lrrrrr}  
					\multicolumn{6}{c}{Dim032 Dataset} \\
					\hline 
					Algorithm & ARI & NMI & CI & TPR & FPR \\ 
					\hline
					OCLUST & 0.99 & 0.99 &   \textbf{0} & \textbf{1.00} & \textbf{0.00}\\ 
					Mean-Shift & \textbf{1.00} & \textbf{1.00} &   \textbf{0} & \textbf{1.00 }& \textbf{0.00} \\ 
					2T & 0.71 & 0.87 &   4 & \textbf{1.00} & \textbf{0.00} \\ 
					CNMix & 0.07 & 0.38 &   1 & 0.99 & 0.60 \\ 
					NCM & 0.98 & 0.99 &   \textbf{0} & \textbf{1.00} & 0.01 \\ 
					DBSCAN & 0.00 & 0.00 &  15 & 0.00 & \textbf{0.00} \\ 
					\hline
			\end{tabular}
		}
		\end{center}		
	\end{table}
	
	For datasets with many clusters, i.e., A2 and A3, OCLUST does better than its competitors. OCLUST and mean-shift have $\text{CI}=0$ for all but one dataset. Mean-shift performs better than the Gaussian model-based approaches as the clusters become highly overlapped in S4. The model-based approaches cannot recover symmetric distributions, such as the Gaussian distribution, for each component when the clusters severely overlap. The OCLUST and mean-shift algorithms attain similar results on the A1, S1, S2, and S3 datasets. Mean-shift performs well, as expected, for those datasets because the clusters are spherical and well-separated.
	With the exception of DBSCAN, most of the algorithms have a nearly zero FPR. DBSCAN has a perfect TPR for four of the datasets, but this is also accompanied by a very large FPR, indicating that DBSCAN is considering most points to be outliers for those datasets. After discarding these results from DBSCAN, OCLUST has the largest TPR for most of the datasets.
	
	OCLUST, mean-shift, and NCM perform very well on the 32-dimensional dataset, {likely due to the uniform generation of outliers.} This is because, as the dimensionality increases, it is more likely for a generated outlier to be far from the clusters in any one direction. OCLUST, CNMix, NCM, and DBSCAN perform well on the Unbalance dataset. Mean-shift filters most of the Unbalance points as noise, resulting in poor classification. DBSCAN performs poorly on most datasets, often over-specifying the number of outliers and combining clusters together. The thresholding approach performs consistently for outlier identification but fails when clusters become more overlapped, likely due to remaining outliers reducing clustering performance. 
	
	In Figure~\ref{fig:KL}, the KL graph for the A3 dataset is plotted. It displays a typical shape, decreasing to a minimum at 446 outliers, after which KL increases. The minimum at 446 is notably different from 525, the number of outliers generated. However, this is to be expected due to the uniform generation of outliers because many of them overlap with the clusters as seen in Figure~\ref{fig:data}. All the algorithms experience the same issue in this respect.
	\begin{figure}[!htb]
		\centering
		\includegraphics[width=4 in]{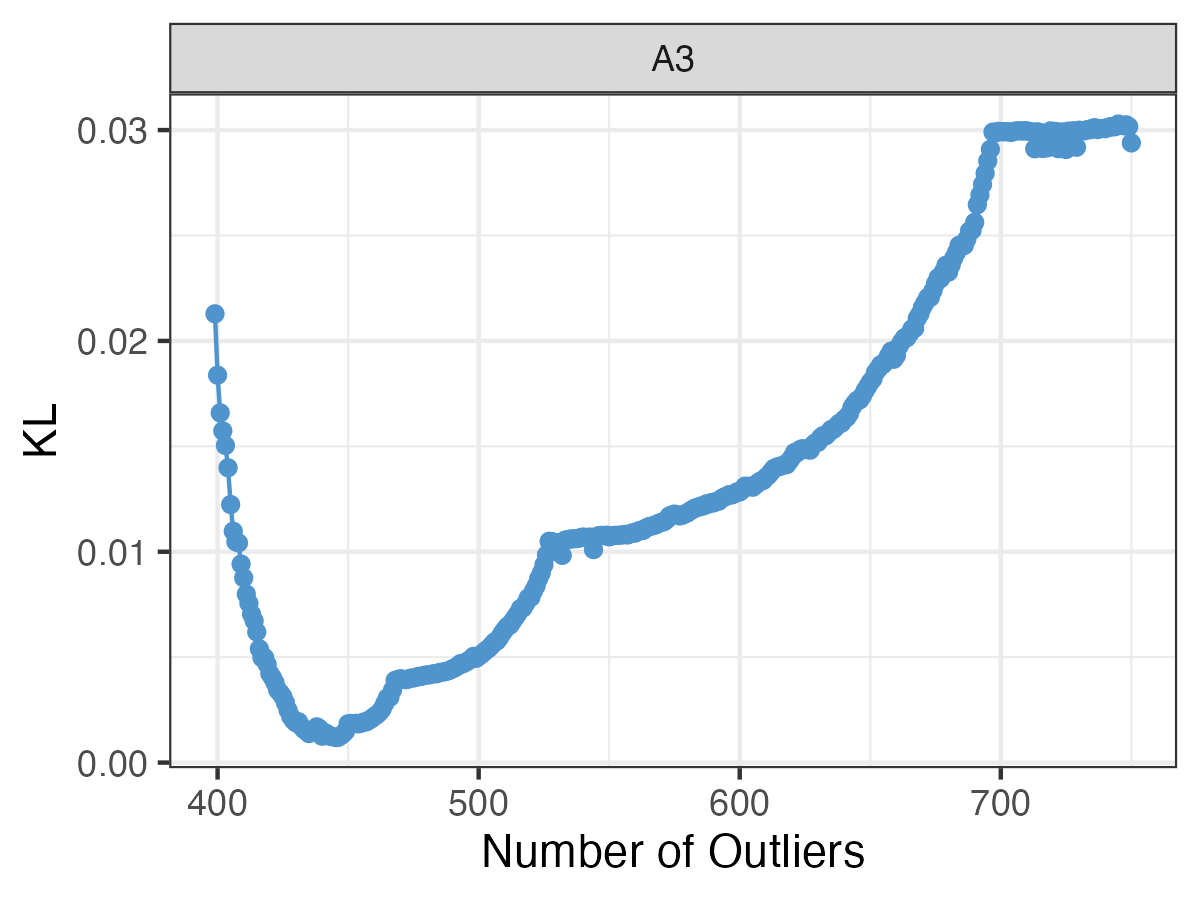}
		\caption{KL plot for the A3 dataset, where the minimum is at 446 outliers.}
		\label{fig:KL}
	\end{figure}
	
	In the initial gross outlier removal for the A3 dataset, 399 outliers are removed. We reach a minimum KL divergence after 48 iterations, but continuing until $F=750$ requires a further 304 iterations. At the minimum KL divergence, we achieve an approximate p-value of 0.099 on Kuiper's test {when $m=100$}. If our significance level was 0.05, we would be able to stop here and substantially save computation time.
	
	\subsection{Crabs Study}\label{sec:crabs}
	Next we evaluate the sensitivity of OCLUST using the  crabs dataset \citep{campbell74}, which is available in the {\tt MASS} package \citep{venables02}. This study closely mimics the study carried out by \cite{peel00} and again by \cite{punzo16b} to demonstrate their respective approaches to dealing with outliers in model-based clustering.
	The dataset contains observations for 100 blue crabs, 50 of which are male, and 50 of which are female. The aim for each classification is to recover the sex of the crab. For this study, we will focus on measurements of rear width (RW) and carapace length (CL). We substitute the CL value of the 25th point to one of eight values in $[-15,20]$. The leftmost plot in Figure~\ref{fig:crabs} plots the crabs dataset by sex, with the permuted value in red taking value $\text{CL}=-5$. We use the OCLUST, mean-shift, CNMix, NCM, and DBSCAN algorithms, along with the thresholding approach. For OCLUST, CNMix and NCM, we run each method, restricting the model to one where the clusters have equal shapes and volumes but varying orientations. We let the \texttt{CNmixt} function determine whether or not the CNMix model is contaminated. Solutions for OCLUST, CNMix, NCM, and the thresholding approach (2T) for the dataset with $\text{CL}=-5$ are  plotted in Figure~\ref{fig:crabs}. Table~\ref{tab:crabs} summarizes the results for each method, listing the number of misclassified points (M) and the predicted number of outliers ($n_O$).  
	\begin{figure}[!htb]
		\centering
		\includegraphics[width=0.9\textwidth]{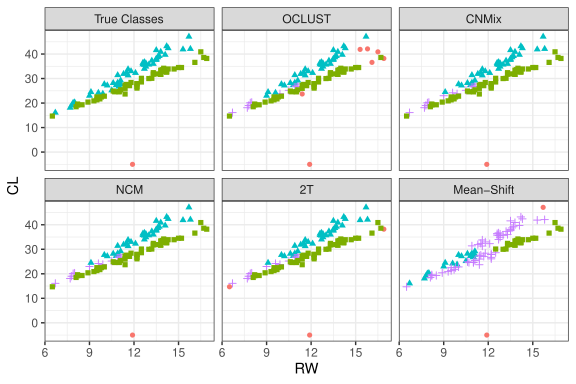}
		\caption{Predicted classifications for OCLUST, CNMix, NCM, mean-shift, and the thresholding approach (2T) when $\text{CL}=-5$. Green squares, blue triangles, purple crosses, and red circles indicate female, male, misclassified, and outlying points, respectively.}
		\label{fig:crabs}
	\end{figure}
	\begin{table}[!htb]
		\caption{Results for running each method on the crabs dataset, where `M' and `$n_O$' designate the number of misclassified points and number of predicted outliers, respectively.} 
		\centering
		\begin{tabular}{@{\extracolsep{\fill}}cccccccccccccccccc}
				\hline
				& \multicolumn{2}{c}{OCLUST} && \multicolumn{2}{c}{Mean-Shift} && \multicolumn{2}{c}{2T}&& \multicolumn{2}{c}{CNMix} && \multicolumn{2}{c}{NCM} && \multicolumn{2}{c}{DBSCAN}\\ 
				\cline{2-3}\cline{5-6}\cline{8-9}\cline{11-12}\cline{14-15}\cline{17-18}
				CL & M & $n_O$  && M & $n_O$  && M & $n_O$ && M & $n_O$&& M & $n_O$&& M & $n_O$\\ 
				\hline
				$-15$ & 11 & 8   && 41 & 3 && 12 & 4  &&  13 & 1 && 13 & 1  &&50& 0 \\ 
				$-10$ & 11 & 6   && 41 & 3  && 13 & 3   && 13 & 1  && 13 & 1  &&50 & 0  \\ 
				$-5$ & 11 & 7   && 41 & 3  && 13 & 3  &&  13 & 1  && 13 & 1  &&50 & 0  \\ 
				0 & 11 & 4  && 41 & 3  && 13 & 2  &&  13 & 1  && 13 & 1 &&50 & 0  \\ 
				5 & 12 & 5   && 41 & 3 && 13 & 3   &&13 & 1  && 13 & 2  &&50 & 0  \\ 
				10 & 11 & 4   && 41 & 3  && 13 & 3   &&  13 & 1  && 11 & 3  &&50 & 0  \\ 
				15 & 11 & 5  && 41 & 3  && 13 & 3   && 13 & 1  && 10 & 4  &&50 & 0  \\ 
				20 & 11 & 5   && 42 & 2  && 14 & 4   &&  13 & 1  && 9 & 5  &&50 & 0  \\				
				\hline 
		\end{tabular} \label{tab:crabs}
	\end{table}
	
	Each method, with the exception of DBSCAN, always identifies the permuted value correctly as an outlier. DBSCAN clusters all the data, including the outlier, into one cluster. NCM classifies more points as outliers as the permuted CL value becomes less extreme. In contrast, OCLUST does the reverse because when the outlier is farther away, it is removed in the gross outlier stage. A gross outlier does not affect the initial clustering, allowing the OCLUST algorithm to remove more points that deviate from multivariate normality.
	
	OCLUST identifies more points as outliers than its competitors and has fewer misclassifications in most instances. It is important to note that the outliers found by OCLUST are not simply the misclassifications --- i.e., outliers classified as non-outliers and \textit{vice versa} --- of NCM. Instead, as seen in Figure~\ref{fig:crabs}, OCLUST identifies three points between the clusters as outliers. This removes the points with high leverage, allowing the clusters to rotate and improve the classification among low values of RW. From one point of view, one may consider that OCLUST removes a few more outliers than necessary; however, their removal improves the parameter estimates. The thresholding approach cannot identify mild outliers because it detects outliers before clustering. Mean shift performs poorly for this dataset due to the shape of the clusters.  
	
	\subsection{Wine Data}
	Finally, we evaluate the results of OCLUST on the wine dataset, available in the {\tt gclus} package \citep{hurley19}. This dataset describes 13 attributes of 178 Italian wines (e.g., alcohol, hue, malic acid). Each wine originates from one of three cultivars: Barolo, Barbera, or Grignolino. For this analysis, we add 12 points of noise, uniformly distributed in each dimension as described in Section~\ref{sec:sims}. We perform unsupervised classification on this dataset with the aim of removing the noise and recovering the cultivar for each wine. We run OCLUST as well as the five comparator methods from Section~\ref{sec:crabs}. The results for the comparators, except DBSCAN, are provided in Table~\ref{tab:winecomp}. DBSCAN performs poorly by classifying all points into one cluster and only identifying seven points of added noise.
	\begin{table}[!htb]
		\centering
		\caption{Classification results for CNMix, NCM, 2T, mean-shift and OCLUST on the wine dataset. Classification results on the wine dataset for OCLUST at global minimum KL and when stopped when the p-value exceeds 0.05.}\label{tab:winecomp}
		%\begin{center}
		\begin{tabular}{@{\extracolsep{\fill}}llccccccccccccc}
				\hline 
				&   \multicolumn{4}{c}{CNMix}&&\multicolumn{4}{c}{NCM}&&\multicolumn{4}{c}{2T}\\ 
				\cline{2-5}\cline{7-10}\cline{12-15}
				{Cultivar}&1&2&3&bad&&1&2&3&bad&&1&2&3&bad\\
				\hline 
				Barbera & 47 && &1 &&47& &  &1&&48&&&\\ 
				Barolo &  &59&& &&&58 &1&&& &59&&\\ 
				Grignolino& 4&&58  & 9&& 3&2  & 65 &1&&5&1&65&\\ 
				Noise& &&  &12 & &  &&  &12&&&&&12\\ 
				\hline
			%	&&&&&&& &  &  && &  && \\ 
				&   \multicolumn{4}{c}{Mean-Shift} &&\multicolumn{4}{c}{OCLUST Min. KL}&&\multicolumn{4}{c}{OCLUST p-val>0.05}\\ 
				\cline{2-5}\cline{7-10}\cline{12-15}
				{Cultivar}&1&2&3&bad&&1&2&3&bad&&1&2&&bad\\
				\hline 
				Barbera &29&&19&&&33& &  &15  && 44&  && 4\\ 
				Barolo &13&44&&2&& & 43& &16 &&  & 58 && 1\\ 
				Grignolino&19&1&50&1&& && 38 &33 && 2& &51&18  \\ 
				Noise&1&&&11&& &&  & 12 && & &&12  \\ 
				\hline
		\end{tabular}
		%\end{center} 
	\end{table}
	
	NCM achieves good results, with only six misclassifications, five of which are from the Grignolino cultivar. It identifies all 12 points of added noise as outliers. CNMix identifies all 12 outliers and has fewer misclassifications than NCM, all four of which are from the Grignolino cultivar. The thresholding approach performs very well, likely due to the uniform simulation of noise in 13 dimensions creating gross outliers. Mean-shift performs poorly. 
	
	The solution for OCLUST, with the minimum KL divergence, is shown in Table~\ref{tab:winecomp}. We have perfect classification among the good points, but 76 points are treated as outliers, which represents 40\% of the dataset. Our KL divergence graph in Figure~\ref{fig:KLwine} has a local minimum at around 35 outliers, {which could be viewed as good alternative solution with fewer outliers.} Re-running the OCLUST algorithm with our p-value stopping criterion results in it stopping at 35 outliers with p-value 0.059. The classification results with 35 outliers are shown in Table~\ref{tab:winecomp}. This allows for the preservation of more original data with minimal decrease in classification accuracy. This alternative solution has fewer misclassifications than its comparators.
	
	The associated KL divergence plot (Figure~\ref{fig:KLwine}) displays a pattern unlike Figure~\ref{fig:KL} because the data are real and not perfectly multivariate normal. We tested each cluster for multivariate normality when all 76 outliers are removed and at the intermediate solution with 35 outliers. We tested the null hypothesis of multivariate normality using the energy test in the \texttt{energy} package \citep{szekely13,szekely22}. When 35 outliers are removed, the Barolo, Barbera, and Grignolino clusters have p-values of 0.09, 0.27, and 0.09, respectively. At the minimum KL solution, i.e., with 76 points removed, the Barolo, Barbera, and Grignolino clusters have p-values of 0.77, 0.35, and 0.59, respectively. Evidently, OCLUST continues to remove points until the clusters are unmistakably multivariate normal. 
	\begin{figure*}[!ht]
		\centering
		\begin{subfigure}[t]{0.5\textwidth}
			\centering
			\includegraphics[width=\textwidth]{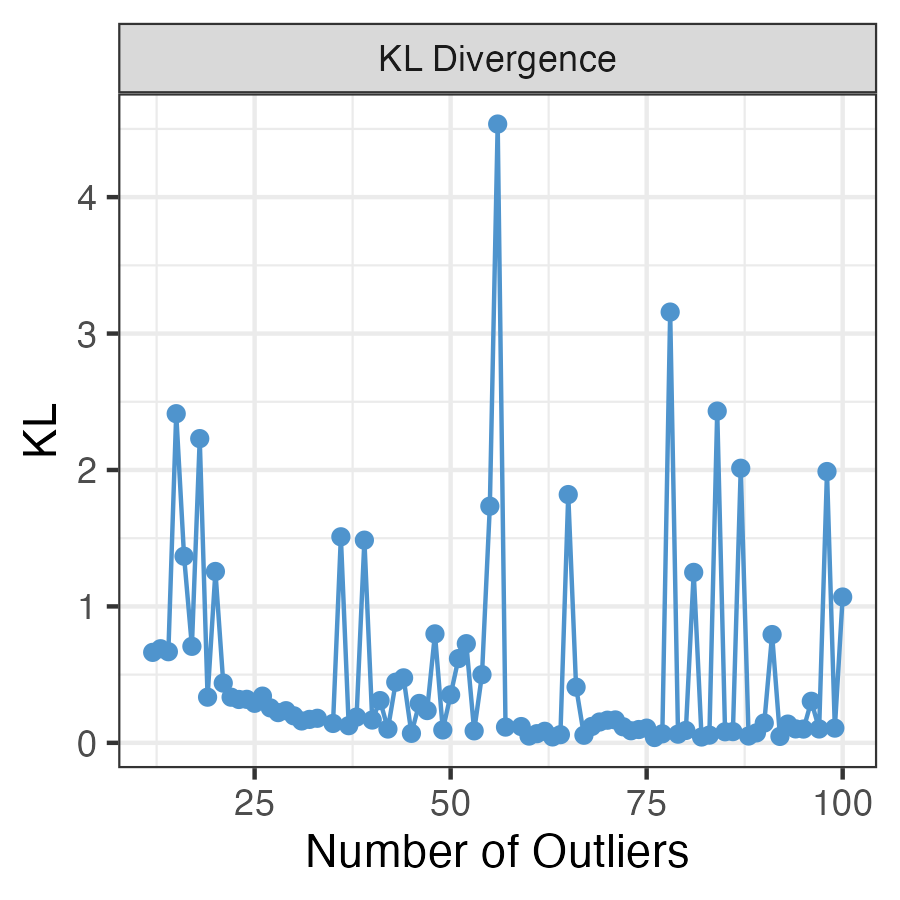}
			\caption{KL plot for the wine dataset. The algorithm is run to an upper bound of 100 outliers.}
			\label{fig:KLwine}
		\end{subfigure}%
		~ 
		\begin{subfigure}[t]{0.5\textwidth}
			\centering
			\includegraphics[width=\textwidth]{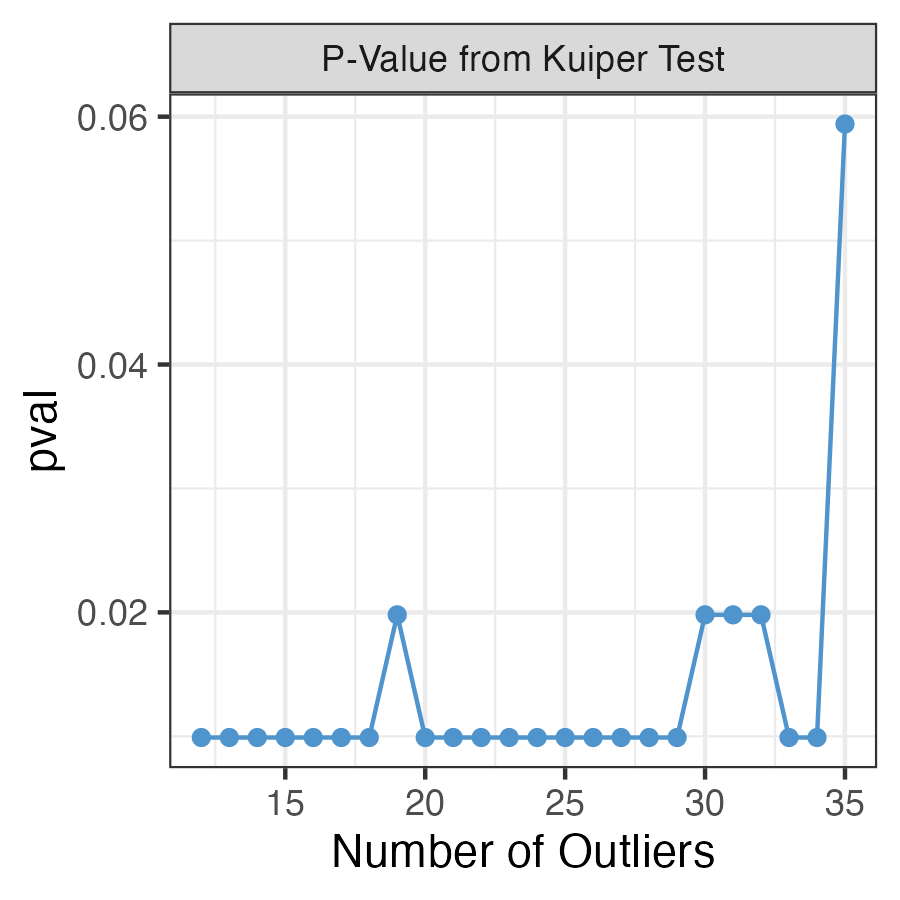}
			\caption{P-value plot for the wine dataset. The algorithm is stopped at 35 outliers because the p-value exceeds 0.05}
			\label{fig:pvalwine}
		\end{subfigure}
		\caption{KL and p-value plots, respectively, for the wine dataset.}
		\label{fig:klpval9}
	\end{figure*}
	
	Two-dimensional projections of both OCLUST solutions are shown in Figure~\ref{fig:wineclass}. In the left-hand figure, where only 35 points are removed, we see some of the identified outliers are far from the clusters but others are mild outliers. In the right-hand plot, we see the solution when more points are removed. Although they may have been correctly classified, they deviate from the symmetric elliptic shape of the multivariate Gaussian distribution. The competing Gaussian model-based classification algorithms, i.e., CNMix and NCM, are able to correctly cluster these points but they seem not to fit well within the model. {Identifying these additional outliers may be useful in novelty detection. These wines present inconsistent features which may be relevant to quality control. The spurious points do not fit into the model and indicate wines with chemical and physical properties that differ from the rest.}
	\begin{figure*}[!ht]
		\centering
		\begin{subfigure}[t]{3in}
			\centering
			\includegraphics[width=3in]{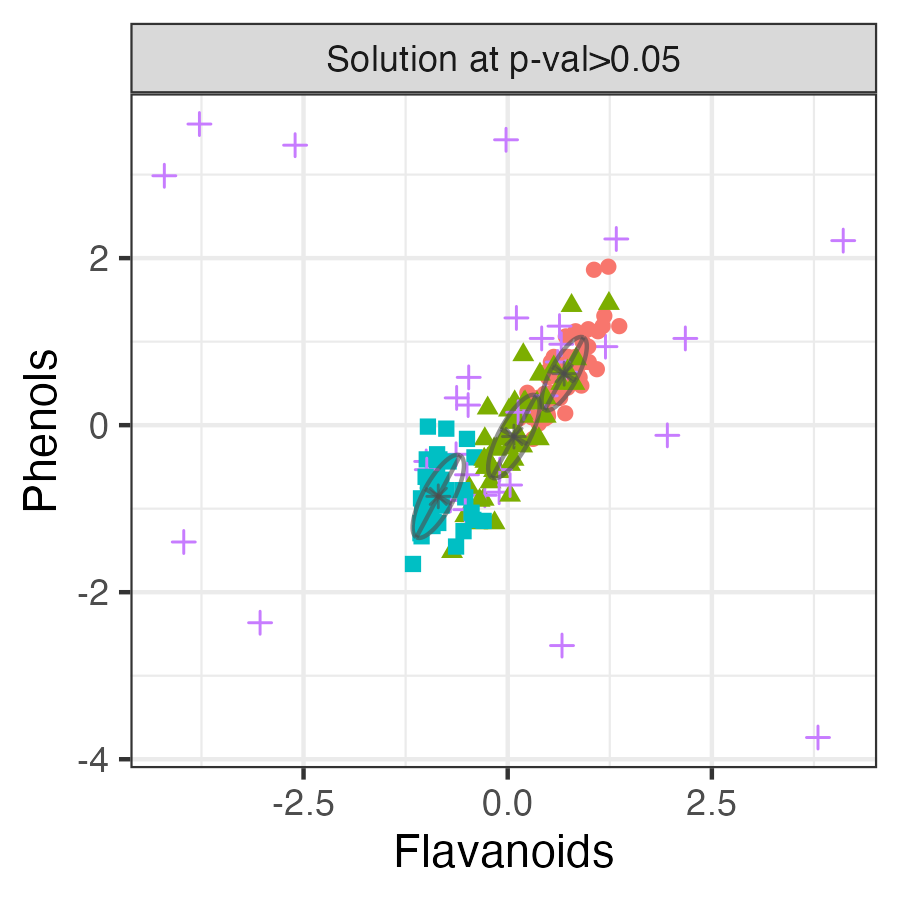}
			\caption{OCLUST classification with 35 identified outliers.}
		\end{subfigure}%
		~ 
		\begin{subfigure}[t]{3in}
			\centering
			\includegraphics[width=3in]{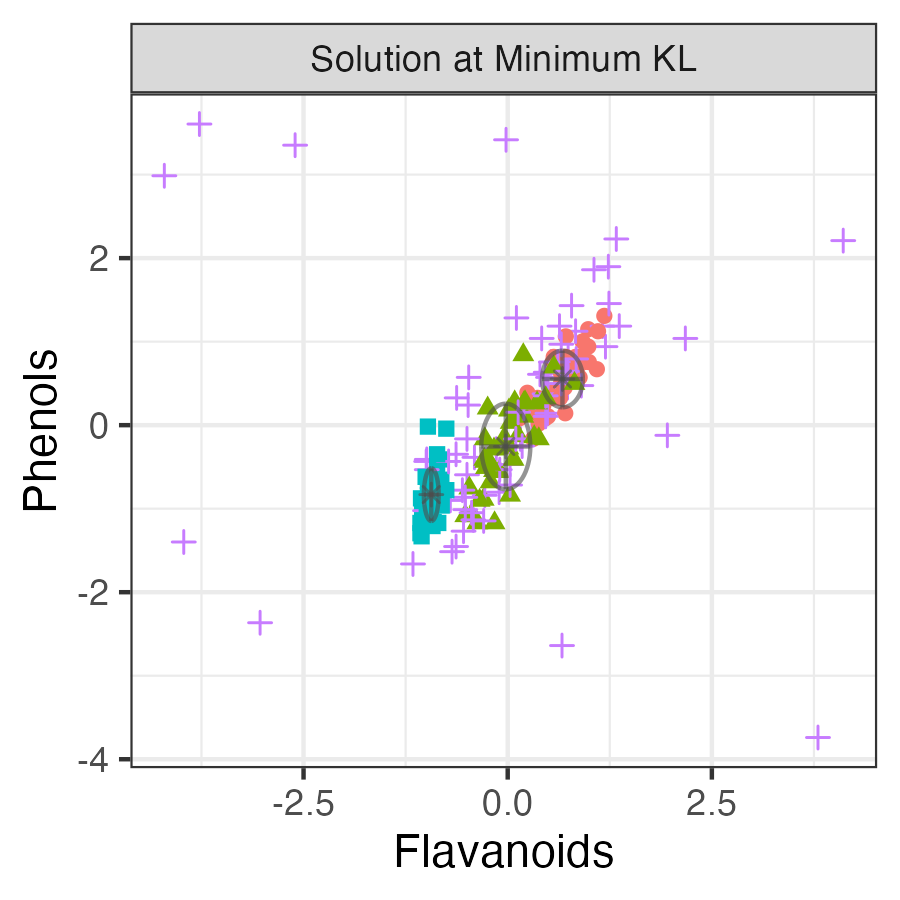}
			\caption{OCLUST classification with 76 identified outliers.}
		\end{subfigure}
		\caption{Predicted classifications by the OCLUST algorithm for the wine dataset with 12 uniform outliers simulated. Scaled data are plotted in two dimensions, showing measurements for flavanoids and phenols. Purple crosses indicate the predicted outliers.}\label{fig:wineclass}
	\end{figure*}
	
	\section{Discussion}
	It was shown that, for data from a finite Gaussian mixture model, the log-likelihoods of the subset models are approximately distributed according to a mixture of beta distributions. This result was used to identify outliers, without needing specification of the proportion of outliers, by removing outlying points until the subset log-likelihoods followed this derived distribution. The result is the OCLUST algorithm, which trims outliers from a dataset while clustering using Gaussian mixture models. In simulations, OCLUST performs similarly or better than mean-shift outlier detection for most benchmark datasets except where there is extreme cluster overlap, i.e., S3 and S4. OCLUST performs similarly to NCM when the clusters are elliptical in the real datasets. For the crabs data, OCLUST trims mild outliers with high leverage, which improves the classification for small values of CL. For the wine dataset, OCLUST obtains good classification with fewer iterations. However, if the algorithm is allowed to continue, it removes points which deviate from the assumption of multivariate normality, arguably removing too many points in the process.
	
	Although this work used the distribution of the log-likelihoods of the subset models to test for the presence of outliers, the derived distribution may be used to verify other underlying model assumptions, such as whether the clusters are Gaussian. Note that the OCLUST algorithm could be used with other clustering methods and should be effective so long as is it reasonable to assume that the underlying distribution of clusters is Gaussian. Of course, one could extend this work by deriving the distribution of subset log-likelihoods for mixture models with non-Gaussian components. This would allow direct consideration of asymmetric clusters with outliers. Finally, this model is limited by dimensionality because, without sufficient observations, the underlying mixture model becomes over parametrized. One could extend this approach to high-dimensional data by using an analogue of the mixture of factor analyzers model or its extensions \citep[see, e.g.,][]{ghahramani97,mcnicholas08,mcnicholas10d}. Based on the comparisons conducted herein, one might expect the resulting method to perform favourably, or at least comparably, when compared to the approaches used by \cite{wei12} and \cite{punzo20}.
	
	\subsection*{Acknowledgments}
%The authors are most grateful to an associate editor and three anonymous reviewers for their helpful comments. 
This work was supported by an NSERC Undergraduate Research Award, an NSERC Canada Graduate Scholarship, the Canada Research Chairs program, an E.W.R. Steacie Memorial Fellowship, and a Dorothy Killam Fellowship.

%\vskip 0.2in
%\bibliographystyle{jasa}
%\bibliography{loglik}

\appendix
\section{Relaxing Assumptions}\label{app:overlap}
Lemma~\ref{lem:loglik} assumes that the clusters are well separated and non-overlapping to simplify the model density to the component density. This appendix, however, serves to show that this assumption may be relaxed in practice. Following \cite{qiu06}, we can quantify the separation between clusters using the separation index $J^*$. In the univariate case,
\begin{equation*}\label{key}
	J^*=\frac{L_2(\alpha/2)-U_1(\alpha/2)}{U_2(\alpha/2)-L_1(\alpha/2)}, 
\end{equation*}
where $L_i(\alpha/2)$ is the sample lower $\alpha/2$ quantile and $U_i(\alpha/2)$ is the sample upper $\alpha/2$ quantile of cluster $i$, and cluster 1 has lower mean than cluster~2. In the multivariate case, the separation index is calculated along the projected direction of maximum separation. Clusters with $J^*>0$ are separated, clusters with $J^*<0$ overlap, and clusters with $J^*=0$ are touching.

To measure the effect of separation index $J^*$ on the complete-data log-likelihood $l_\mathcal{X}$, 100 random datasets with $n=1800$ for each combination of settings used were generated using the {\tt clusterGeneration} \citep{qiu15} package in {\sf R}. The settings used were: three clusters with equal cluster proportions, dimensions $p \in \{2,4,6\}$, and separation indices in $[-0.9,0.9]$.  Covariance matrices were generated using random eigenvalues $\lambda\in(0,10]$. 
The parameters were estimated using the EM algorithm. The log-likelihood and the complete-data log-likelihood were calculated using the parameter estimates. The average of the quantity $(l_\mathcal{X}-\ell_\mathcal{X})/\ell_\mathcal{X}$ is computed for each combination of settings --- recall that there are 100 datasets for each setting. The result is the average proportional change in log-likelihood over the 100 datasets between the log-likelihood $\ell_\mathcal{X}$ and the complete-data log-likelihood $l_\mathcal{X}$ (Figure~\ref{fig:sepvslik}).
\begin{figure}[!htb]
	\centering
	\includegraphics[width=3.5in]{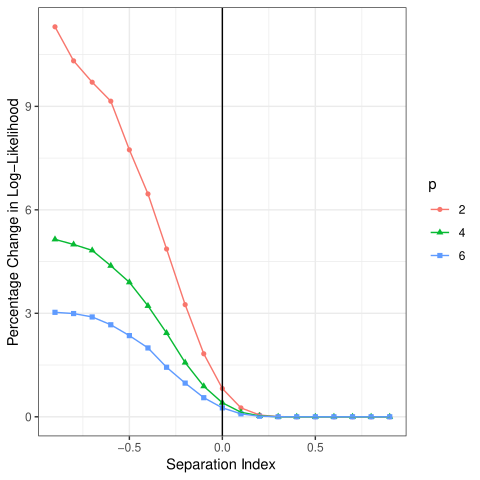}
	\caption{The effect of cluster separation on the complete-data log-likelihood of the model, where the vertical line represents the threshold between separated and overlapping clusters.}
	\label{fig:sepvslik}
\end{figure}

As one would expect, the approximation of $\ell_\mathcal{X}$ by $l_\mathcal{X}$ improves as the separation index increases (see Lemma~\ref{lem:loglik}). %and its proof in Appendix~\ref{app:proofs}). 
However, the difference is negligible for touching and separated clusters ($J^*\geq0$).

%\newpage
\section{Mathematical Results}\label{app:proofs}
\subsection{Proof of Lemma~\ref{lem:loglik}}
\begin{proof}
	Suppose $\vecSigma$ is positive definite. Then, $\vecSigma^{-1}$ is also positive definite and there exists a matrix $\vecQ$ such that $\vecQ'\vecQ=\mathbf{I}$ and $\vecSigma^{-1}=\vecQ'\vecLambda\vecQ$, where $\vecLambda$ is diagonal with $\vecLambda_{ii}=\lambda_i>0, i\in[1,p]$. Let  $\vecx-\vecmu%=\vecy
	=\vecQ'\vecw$, where $\vecw\neq \mathbf{0}$. Now,	
	\begin{equation*}
		\begin{split}
			(\vecx-\vecmu)'\vecSigma^{-1}(\vecx-\vecmu)&%=\vecy'\vecSigma^{-1}\vecy
			=\vecw'\vecQ\vecSigma^{-1}\vecQ'\vecw%=\vecw'\vecQ\vecQ'\vecLambda\vecQ\vecQ'\vecw
			=\vecw'\vecLambda\vecw
			=\sum_{i=1}^p \lambda_i w_i^2\\
			&\geq \inf_i(\lambda_i) \sum_{i=1}^p  w_i^2
			=\inf_i(\lambda_i) \|\vecw\|^2
			=\inf_i(\lambda_i) \|\vecx-\vecmu\|^2
		\end{split}
	\end{equation*}
	because $\|\vecx-\vecmu\|^2%=\|\vecy\|^2
	=\|\vecQ'\vecw\|^2%=(\vecQ'\vecw)'\vecQ'\vecw
	=\vecw'\vecQ\vecQ'\vecw%=\vecw'\vecw
	=\|\vecw\|^2$.
	Thus, as $\|\vecx-\vecmu\|\rightarrow \infty$, $$(\vecx-\vecmu)'\vecSigma^{-1}(\vecx-\vecmu) \rightarrow \infty$$ and
	\begin{equation*}
		\phi(\vecx\mid\vecmu,\vecSigma)=\frac{1}{\sqrt{(2\pi)^p|\vecSigma|}}\text{exp}\left\{-\frac{1}{2}(\vecx-\vecmu)'\vecSigma^{-1}(\vecx-\vecmu)\right\} \rightarrow 0.
	\end{equation*}
	Suppose $z_{ih}=1$. Then, as the clusters separate, $ \|\vecx_i-\vecmu_g\| \rightarrow \infty$ and $\phi(\vecx_i\mid \vecmu_g,\vecSigma_g) \rightarrow 0$ for $g\neq h$. Thus, for $\vecx_i$,
	\begin{equation*}
		\begin{split}
			\sum_{g=1}^G \pi_g \phi(\vecx_i\mid \vecmu_g, \vecSigma_g) &= \sum_{g\neq h}  \pi_g \phi(\vecx_i\mid \vecmu_g, \vecSigma_g)
			+ \pi_h \phi(\vecx_i\mid \vecmu_h, \vecSigma_h) \rightarrow \pi_h \phi(\vecx_i\mid \vecmu_h, \vecSigma_h)
		\end{split}
	\end{equation*}
as $ \|\vecx_i-\vecmu_g\| \rightarrow \infty, g\neq h$. Therefore,
	\begin{equation*}
		\begin{split}
			\ell_{\mathcal{X}}&=\sum_{i=1}^n \log \left[\sum_{g=1}^G \pi_g \phi(\vecx_i\mid \vecmu_g, \vecSigma_g)\right]
			\rightarrow \sum_{i=1}^n\sum_{g=1}^G  z_{ig} \log \left[\pi_g \phi(\vecx_i\mid \vecmu_g, \vecSigma_g)\right]=l_\mathcal{X}
		\end{split}
	\end{equation*}
as the clusters separate. 
\end{proof}
\begin{rmk}
	{Although covariance matrices need only be positive semi-definite, we restrict $\vecSigma$ to be positive definite in the proof of Lemma~\ref{lem:loglik} so that $\vecX$ is not degenerate.} 
\end{rmk}

\subsection{Proof of Parameter Estimate Convergence}\label{app:estconv}
{Sample parameter estimates for the subset models converge to the parameter estimates for the full model. That is, as $n_g\rightarrow\infty$:}
\begin{equation*}
	\begin{split}
		\hat{\pi}_{g\setminus j}\rightarrow \hat{\pi}_g,\\
		\bar{\vecx}_{g\setminus j}\rightarrow \bar{\vecx}_g,\\
		\sampcov_{g\setminus j}\rightarrow\sampcov_g,
	\end{split}
\end{equation*}
where $\hat{\pi}_{g}$, $\bar{\vecx}_g$, and $\sampcov_g$ are the sample proportion, mean, and covariance, respectively, for the $g${th} cluster {considering all observations in the entire dataset $\mathcal{X}$,} and $\hat{\pi}_{g\setminus j}$, $\bar{\vecx}_{g \setminus j}$, and $\sampcov_{g\setminus j}$ are the sample proportion, mean, and covariance, respectively, for the $g${th} cluster considering only observations in the $j${th} subset $\mathcal{X}\setminus \vecx_j$.

\begin{proof}
	If $z_{jh}=1$, then 
	\begin{equation*}
		\hat{\pi}_{h \setminus j}=\frac{n_h-1}{n-1}.\\
	\end{equation*}
	If $z_{jh}=0$,
\begin{equation*}
	\hat{\pi}_{h \setminus j}=\frac{n_h}{n-1}.\\
\end{equation*}
	Thus, $$\hat{\pi}_{h \setminus j}\rightarrow \frac{n_h}{n}=\hat{\pi}_{h}$$ as $n_h\rightarrow \infty$.
	Now, for $\bar{\vecx}_{h \setminus j}$ and $\sampcov_{h \setminus j}$, if $z_{jh}=0$, then $\bar{\vecx}_{h \setminus j}=\bar{\vecx}_{h}$ and $\sampcov_{h\setminus j}=\sampcov_{h}$ because
	\begin{equation*}
n_{h \setminus j}=n_h, \qquad \bar{\vecx}_h=\frac{1}{n_h}\sum_{i=1}^n z_{ih}\vecx_i= \frac{1}{n_h}\sum_{i\neq j} z_{ih}\vecx_i =\bar{\vecx}_{h \setminus j},
	\end{equation*}
\begin{equation*}
\sampcov_h=\frac{1}{n_h-1}\sum_{i=1}^n z_{ih}(\vecx_i-\bar{\vecx}_h)(\vecx_i-\bar{\vecx}_h)'=\frac{1}{n_h-1}\sum_{i\neq j} z_{ih}(\vecx_i-\bar{\vecx}_h)(\vecx_i-\bar{\vecx}_h)'=\sampcov_{h\setminus j} .
\end{equation*} \\
	If $z_{jh}=1$,
	\begin{equation*}
		\bar{\vecx}_{h \setminus j}=\frac{n_h\bar{\vecx}_h-\vecx_j}{n_h-1}.\\
	\end{equation*}
	Thus, $\bar{\vecx}_{h \setminus j}\rightarrow \bar{\vecx}_h$ as $n_h\rightarrow \infty$.
	%\textcolor{red}{The sample covariance matrix  $\sampcov_g$ requires the above result, that $\bar{\vecx}_{h \setminus j}\simeq\bar{\vecx}_h$.
		Therefore, $\bar{\vecx}_{h \setminus j}\approx\bar{\vecx}_h$ and
		\begin{equation*}
			\sampcov_{h \setminus j}\approx\frac{1}{n_h-2}\left[(n_h-1)\sampcov_h-(\vecx_j-\bar{\vecx}_h)(\vecx_j-\bar{\vecx}_h)'\right].
		\end{equation*}
		Thus, $\sampcov_{h\setminus j} \rightarrow \sampcov_h$ as $n_h\rightarrow \infty$.
	\end{proof}
%\newpage
\section{Computation Time}\label{app:timing}
Table~\ref{tab:time} records the computation time in seconds for each algorithm. The entire OCLUST algorithm is timed without the p-value stopping criterion. The mean-shift algorithm is run in Python with 5000 random swaps. All other algorithms are run in \textsf{R}. Competing algorithms are run in series on a cluster with Intel Xeon Silver 4114 CPUs. The OCLUST algorithm was run on that same cluster for all datasets except Dim032. The benchmark and crabs datasets were run in parallel with 80 cores and the wine dataset was run with 20 cores. Dim032 was run with 30 cores on a separate cluster with Intel Xeon E5-4627 v2 CPUs.  
\begin{table}[!htp]
	\caption{Computation time in seconds for each outlier identification algorithm on all datasets.}\label{tab:time}
	\resizebox{\textwidth}{!}{\begin{tabular}{lcccccc}  
			\hline 
			Dataset & (Full) OCLUST & Mean-Shift & 2T & CNMix & NCM & DBSCAN \\   
			\hline
			A1 & 5.43E+03 & 2.98E+04 & 2.38E+02 & 2.77E+02 &   6.14E+00& 2.92E-01\\ 
			A2 & 3.59E+04 & 6.42E+04 & 7.29E+02 & 1.44E+02 &   9.86E+00& 2.46E-01\\ 
			A3 & 1.21E+05 &  2.01E+05& 2.11E+03 & 5.16E+03 &   3.71E+01 &4.70E-01\\ 
			S1& 1.03E+04 & 4.28E+04 & 1.15E+02 & 1.84E+02 &   6.76E+00&2.69E-01 \\ 
			S2 & 1.75E+04 & 4.06E+04 & 2.14E+02 & 3.06E+02 &   1.03E+01 &2.60E-01\\ 
			S3 & 6.10E+04 & 3.13E+04 & 2.86E+02&  3.50E+02&  1.96E+01& 3.07E-01\\ 
			S4 & 1.27E+05 & 2.47E+04 & 3.60E+02 & 4.65E+02 &  2.19E+01& 2.22E-01\\ 
			Unbalance & 9.33E+05 &1.96E+04 & 3.80E+02 &1.26E+02 &  5.53E+00&5.67E-01 \\ 
			Dim032 & 7.35E+03 & 8.566E+03 & 2.28E+02 & 1.55E+03 &  7.71E+01&3.05E-01 \\
			Wine & 3.95E+02 & 1.80E+02 &2.38E+02 &  5.62E+00&  2.69E-01& 7.09E-02\\  
			Crabs (avg.) & 3.65E+00 & 1.06E+02 & 2.93E-02 & 9.02E-01 &  9.66E-02& 3.66E-02\\ 
			\hline
	\end{tabular}}
\end{table}

\end{document}